\documentclass[12pt,a4paper]{article}
\usepackage{authblk}
\usepackage[top=2.5cm, bottom=2.5cm, left=2cm, right=2cm]{geometry}

\usepackage[T1]{fontenc}
\usepackage[utf8]{inputenc}
\usepackage{mathpazo}          
\usepackage{microtype}         

\usepackage[bottom]{footmisc}

\usepackage{enumitem}

\usepackage{booktabs}
\usepackage{multirow}
\usepackage{blkarray}
\usepackage{xcolor,colortbl}

\usepackage{graphicx}
\usepackage{epsfig}
\usepackage{tikz}
\usetikzlibrary{calc,matrix,through,shapes.geometric,shapes.misc}

\usepackage[normalem]{ulem}
\usepackage{lipsum}
\usepackage{verbatim}
\usepackage{ragged2e}

\usepackage{amsmath}
\usepackage{amsfonts}
\usepackage{amssymb,amsthm}
\usepackage{latexsym}
\usepackage{mathrsfs}
\usepackage{braket}
\usepackage{gensymb}

\usepackage{caption}
\usepackage{subcaption}
\DeclareCaptionJustification{justified}{\justifying}
\usepackage{setspace}
\usepackage{titlesec}
\titleformat{\section}{\large\bfseries}{\thesection.}{0.8em}{}
\titleformat{\subsection}{\normalsize\bfseries}{\thesubsection.}{0.6em}{}
\titleformat{\subsubsection}{\normalsize\itshape}{\thesubsubsection.}{0.5em}{}

\usepackage{cite}

\usepackage[colorlinks=true,linkcolor=blue,citecolor=blue,urlcolor=blue]{hyperref}

\allowdisplaybreaks

\newtheorem{theorem}{Theorem}
\newtheorem{lemma}{Lemma}

\newtheorem{proposition}{Proposition}
\newtheorem{observation}{Observation}

\newtheorem{remark}{Remark}

\newtheorem{question}{Question}
\newtheorem{definition}{Definition}

\newtheorem*{task}{Task}

\usepackage{authblk}

\title{\textbf{Conclusive Identification Via Noisy Classical Channel: Superactivation and Quantum Advantage}}

\author[1]{Anushko Chattopadhyay}
\author[1]{Ambuj}
\author[1]{Rakesh Das}
\author[1]{Smritikana Patra}
\author[1]{Chitrak Roychowdhury}
\author[2]{Manik Banik}
\author[1]{Amit Mukherjee}

\affil[1]{Indian Institute of Technology Jodhpur, Jodhpur 342030, India}
\affil[2]{S.\,N.\,Bose National Centre for Basic Sciences, Block JD, Sector III, Salt Lake, Kolkata 700106, India}

\date{}

\begin{document}

\maketitle

\begin{abstract}
We introduce and study the \emph{conclusive identification} task for classical channels, in which a receiver identifies transmitted inputs without error whenever possible, while responding inconclusively when the output is ambiguous. For a symmetric not-fully-corrupted channel $N : X \to X$, the single-shot conclusive identification index $\mathrm{ci}_\circ(N)$ counts the maximum number of conclusively identifiable inputs. We show that $\mathrm{ci}_\circ(N)$ can exhibit a striking superactivation phenomenon: a channel with $\mathrm{ci}_\circ(N) = 0$ achieves $\mathrm{ci}_\circ(N \otimes \mathrm{id}^c_\beta) = |X|$ when assisted by a perfect classical channel $\mathrm{id}^c_\beta$ of dimension $\beta < |X|$; and the minimum classical assistance required is precisely the chromatic number $\chi(\mathtt{S}_N)$ of the support graph $\mathtt{S}_N$ of $N$. We provide explicit channel families for which superactivation gap $\mathrm{ci}_\circ(N \otimes \mathrm{id}^c_\beta) - \mathrm{ci}_\circ(\mathrm{id}^c_\beta)$ can be made arbitrarily large. We further show that a noiseless quantum channel of dimension equal to the orthogonal rank $\xi(\mathtt{S}_N)$ of the support graph suffices the purpose, establishing a strict quantum advantage whenever $\xi(\mathtt{S}_N) < \chi(\mathtt{S}_N)$. This advantage is demonstrated through three explicit constructions motivated from combinatorial state-independent, algebraic state-independent, and state-dependent proof of Kochen-Specker contextuality. Moreover, using the co-normal product of graphs, we show the scaling of quantum advantage ratio $\chi_f(\mathtt{S}_N)/\xi(\mathtt{S}_N)$, and also present example of channel for which quantum assistance can be exponentially more efficient than classical assistance. Taken together, our results establish the support graph $\mathtt{S}_N$, rather than the confusability graph $\mathtt{G}_N$, as the appropriate combinatorial object for characterizing channel utility in the conclusive identification task, and that channels dismissed as entirely useless by Shannon's zero-error framework can exhibit rich superactivation phenomenon and quantum advantage, while finding deep connections to quantum contextuality.
\end{abstract}

\section{Introduction}\label{sec:intro}

\noindent The capacity of a noisy classical channel to transmit information reliably is one of the central questions of information theory. Shannon's landmark work~\cite{Shannon1948} established the transmission capacity $C(N)$ as the supreme rate at which information can be sent with vanishingly small error probability over a discrete memoryless channel (DMC) $N : X \to Y$. A more stringent variant, the \emph{zero-error capacity} $C_\circ(N)$, introduced by Shannon~\cite{Shannon1956}, requires the probability of error to be exactly zero rather than merely vanishing. The zero-error capacity depends solely on the combinatorial structure of the channel through its \emph{confusability graph} $\mathtt{G}_N$, whose vertices are the channel inputs and whose edges connect pairs of inputs that can produce a common output. Specifically, single-shot zero-error capacity $c_\circ(N) = \log_2\alpha(\mathtt{G}_N)$, where $\alpha(\mathtt{G}_N)$ is the independence number of $\mathtt{G}_N$, and $C_\circ(N) = \lim_{k \to \infty} \frac{1}{k} \log_2 \alpha\big(\mathtt{G}_N^{\boxtimes k}\big)$, where $\mathtt{G}_N^{\boxtimes k}$ is the $k$-fold strong graph product of $\mathtt{G}_N$ with itself.

\noindent The zero-error framework has led to some of the most beautiful results in information theory and combinatorics \cite{Csiszar2015}. Shannon~\cite{Shannon1956} showed that the zero-error capacity can be strictly superadditive: there exist channels $N$ for which $c_\circ(N^{\otimes k}) > k \cdot c_\circ(N)$. The quintessential example is the channel $N_5$ whose confusability graph is the cycle $\mathtt{C}_5$, for which $\alpha(\mathtt{C}_5) = 2$ yet $\alpha\big(\mathtt{C}_5^{\boxtimes 2}\big) = 5$. Lov\'{a}sz's celebrated result~\cite{Lovasz1979} resolved the capacity of $N_5$ by introducing the theta function $\vartheta(\mathtt{G}_N)$, proving $C_\circ(N_5) = \frac{1}{2}\log_2 5$. The role of quantum resources in zero-error communication has also been studied extensively \cite{Guedes2016}. Cubitt et al.~\cite{Cubitt2010} demonstrated that shared entanglement can strictly increase the zero-error capacity of certain channels, while a more recent work~\cite{Ambuj2026} showed that a noisy classical channel paired with a perfect quantum channel can exhibit capacity enhancements beyond what classical assistance alone can achieve.

\noindent A recurring theme in these results is that the confusability graph $\mathtt{G}_N$ is the central combinatorial object governing zero-error communication. However, $\mathtt{G}_N$ discards finer characteristic features of the channels: many distinct channel matrices can share the same confusability graph, and thus appear identical from the zero-error perspective. In particular, any $n$-input channel $N$ with $\mathtt{G}_N = \mathtt{K}_n$, the complete graph on $n$ vertices, satisfies $C_\circ(N) = 0$, regardless of the finer structure of $N$. Such channels are deemed entirely useless for zero-error communication, and by extension for the identification framework later introduced by Ahlswede and Dueck~\cite{Ahlswede1989-1}.

\subsection*{Main Contributions}

\noindent In this paper, we ask: \emph{are channels with $C_\circ(N) = 0$ truly useless for all communication tasks, in particular within single-shot regime?} We answer this question in the negative by introducing a new communication task -- \emph{conclusive identification} -- and demonstrating that channels useless in the zero-error framework can exhibit rich and previously unobserved superactivation and quantum advantage phenomena.\\

\noindent$\checkmark${\it The Conclusive Identification Task}\vspace{-.30cm}\\

\noindent Given a classical channel $N : X \to X$ from a sender (Alice) to a receiver (Bob), the conclusive identification task requires Bob to identify the transmitted message $x \in X$ without error whenever he is certain of the answer, while being permitted to respond \emph{inconclusively} when the output is ambiguous. This task interpolates between the zero-error communication framework -- where Bob must always commit to a decision -- and the trivial task where Bob always responds inconclusively. The key metric is the \emph{conclusive identification index} $\mathrm{ci}_\circ(N)$, which counts the maximum number of inputs that Bob can conclusively and correctly identify.

\noindent The conclusive identification task is related to, but fundamentally distinct from, the identification framework of Ahlswede and Dueck~\cite{Ahlswede1989-1}. In the latter, two types of errors are considered: the error of the first kind (missed detection), where the receiver rejects the true message, and the error of the second kind (false acceptance), where an incorrect message is accepted. The Ahlswede-Dueck formulation permits both error probabilities to vanish asymptotically. In contrast, conclusive identification operates in the zero-error regime with the stringent requirement that the error of the second kind vanishes exactly ($\lambda_2 = 0$), while allowing a nonzero probability of the first kind ($\lambda_1 \leq 1$). This formulation naturally accommodates an inconclusive outcome, analogous to unambiguous state discrimination in quantum information theory \cite{Chefles2000,Barnett2009,Bae2015}, where the measurement may abstain rather than risk an incorrect decision.\\

\noindent$\checkmark${\it Support Graph and Graph Coloring}\vspace{-.3cm}\\

\noindent A central finding of this paper is that the appropriate combinatorial object for conclusive identification is not the confusability graph $\mathtt{G}_N$, but rather the \emph{support graph} $\mathtt{S}_N$, whose edges encode the zero-nonzero pattern of the channel matrix. For a symmetric not-fully-corrupted (SNFC) channel, $\mathtt{S}_N$ is an undirected graph with self-loops (often loops will not be shown for simplicity of representation), and we establish the following fundamental result:

\begin{quote}
\emph{The minimum classical assistance $\mathrm{id}^c_\beta$ required to conclusively identify all $|X|$ inputs of an SNFC channel $N$ having $c_\circ(N)=0~\&~\mathrm{ci}_\circ(N)=0$ is precisely the chromatic number $\chi(\mathtt{S}_N)$ of the support graph.}
\end{quote}

\noindent This forges a direct and unexpected connection between classical channel capacity theory and the graph coloring problem. In particular, the partition of the input alphabet $X$ into color classes corresponds exactly to Alice's encoding strategy for the assisted classical channel, and the properness of the coloring guarantees that inputs in the same class are conclusively distinguishable by Bob.\\

\noindent$\checkmark${\it Superactivation}\vspace{-.3cm}\\

\noindent We demonstrate that the conclusive identification index exhibits a striking \emph{superactivation} phenomenon. There exist SNFC channels $N$ with $c_\circ(N) = 0$ (completely useless for zero-error communication) and $\mathrm{ci}_\circ(N) = 0$ (completely useless for unassisted conclusive identification) yet 
\begin{equation*}
\mathrm{ci}_\circ(N \otimes \mathrm{id}^c_\beta) = |X|, \qquad \beta = \chi(\mathtt{S}_N) < |X|,
\end{equation*}
so that joint use of $N$ with a strictly smaller classical channel achieves perfect conclusive identification of all inputs. As we show the superactivation gap $|X| - \beta$ can be made arbitrarily large: for channels with support graph $\mathtt{W}_n$ (wheel graph), $\mathtt{F}_n$ (friendship graph), $\mathtt{St}_n$ (star graph), and $\mathtt{T}_{2,n}$ (Turán graph) -- the gap grows as $n - O(1)$, demonstrating unbounded superactivation. This contrasts sharply with the zero-error framework, where no such superactivation is possible when $c_\circ(N) = 0$.\\

\noindent$\checkmark${\it Quantum Advantage}\vspace{-.3cm}\\

\noindent We further investigate quantum assistance in the conclusive identification task and establish the following result:
\begin{quote}
\emph{Assistance of a noiseless quantum channel $\mathrm{id}^q_\beta$ is sufficient to conclusively identify all $|X|$ inputs of an SNFC channel $N$ having $c_\circ(N)=0~\&~\mathrm{ci}_\circ(N)=0$, where $\beta$ is the (complex) orthogonal rank $\xi(\mathtt{S}_N)$ of the support graph.}
\end{quote}
It thus exhibits a strict \emph{quantum advantage} whenever $\xi(\mathtt{S}_N) < \chi(\mathtt{S}_N)$. We demonstrate this advantage by exploiting structural features rooted in the Kochen--Specker (KS) theorem~\cite{Kochen1967}. We substantiate this advantage via three explicit constructions of increasing generality: combinatorial state-independent contextuality, algebraic state-independent contextuality, and state-dependent contextuality. These results demonstrate that the advantage is not tied to a particular form of contextuality proof, but rather follows from the separation $\xi(\mathtt{S}_N) < \chi(\mathtt{S}_N)$. Finally, we show that the quantum advantage ratio $\mathtt{QA}(N) := \frac{\chi(\mathtt{S}_N)}{\xi(\mathtt{S}_N)}$ can be made arbitrarily large under the co-normal product of graphs, thereby establishing scalability of the quantum advantage, and also provide example of channel for which quantum assistance is exponentially better than its classical counterpart. \vspace{-.25cm}\\

\noindent$\clubsuit$ {\it Organization}\vspace{-.25cm}\\

\noindent The remainder of the paper is organized as follows. Section \ref{sec2} provides a brief overview of zero-error communication theory. In Section \ref{sec3}, we first formally introduce the conclusive identification task and the support graph framework. We also present the results with the assistance of noiseless classical channel and demonstrate the striking superactivation phenomenon. 
In Section \ref{sec4}, we establish a strict quantum advantage, showing that the orthogonal rank of the support graph, which can strictly undercut the chromatic number, is the corresponding quantity for quantum assistance. We construct explicit channel families motivated by combinatorial state-independent, algebraic state-independent, and state-dependent KS contextuality proofs, and show that the quantum advantage can be scaled arbitrarily, with the ratio of classical to quantum resources growing exponentially in the natural complexity parameter of the channel.\\

\noindent$\clubsuit$ {\it Related Works}\vspace{-.25cm}\\

\noindent The zero-error capacity of classical channels was introduced by Shannon~\cite{Shannon1956} and remains an active area of research~\cite{Lovasz1979,Haemers1979}. The role of quantum resources in zero-error communication was investigated by Cubitt et al.~\cite{Cubitt2010} and extended in subsequent work~\cite{Leung2012}. The identification framework of Ahlswede and Dueck~\cite{Ahlswede1989-1} established the doubly exponential scaling of identification codes and its equality with the Shannon capacity. Graph coloring and chromatic numbers have appeared in information theory in the context of index coding, network coding, and zero-error source coding~\cite{Witsenhausen1976,Korner1988,BarYossef2006,Effros2015,Zuiddam2019}. The quantum chromatic number and its connection to KS contextuality were introduced by Cameron et al.~\cite{Cameron2007} and developed further in~\cite{Paulsen2016}. To the best of our knowledge, the conclusive identification task, its connection to graph coloring, and the superactivation and quantum advantage phenomena established here are new.

\section{Overview}\label{sec2}

\noindent A classical communication channel is a medium that carries classical information from a sender to a receiver. Mathematically, such a channel is described by a conditional probability distribution that maps channel inputs to outputs \cite{Cover2005}. Let $X$ and $Y$ denote the input and output alphabets, respectively, which we assume to be finite throughout this work.

\begin{definition}[Classical Channel]\label{def1}
A discrete memoryless classical channel (DMC) $N: X\to Y$ is specified by a stochastic matrix $\{P(y|x)\ge0~|~\sum_yP(y|x)=1\}$; $|X|,~|Y|<\infty$.
\end{definition}
\noindent Such a channel $N$ is associated with a  channel matrix (also called stochastic matrix)  $\mathbb{N}=(N_{yx})\in\mathbb{R}^{|Y|\times|X|}$, with entries given by $N_{yx}= P(y|x)$ for $x\in X~\&~y\in Y$. Thus, each column of $\mathbb{N}$ represents a probability vector in $\mathbb{R}^{|X|}$. 

\begin{definition}[Confusability Graph]\label{def2}
The confusability graph of a DMC $N$ is the simple graph $\mathtt{G}_N=(V,E)$ with vertex set $V=X$, where two distinct inputs $x,x'\in X$ are adjacent if they can produce a common output. 
\end{definition}
\noindent The confusability graph plays a central role in the study of zero-error communication over classical channels.
\begin{definition}[Single-shot zero-error capacity]\label{def3}
Let $\mathtt{G}_N$ denote the confusability graph of a classical channel $N$. The \emph{single-shot zero-error capacity} of $N$, measured in bits, is defined as $c_\circ(N)=\log_2 \alpha(\mathtt{G}_N)$, where $\alpha(\mathtt{G}_N)$ denotes the independence number of $\mathtt{G}_N$, i.e., the maximum cardinality of a subset of vertices in $\mathtt{G}_N$ containing no pair of adjacent vertices.
\end{definition}
\noindent When multiple independent uses of a channel $N$ are considered, the confusability graph of the $n$-fold channel $N^{\otimes n}$ is given by the $n$-fold strong graph product $\mathtt{G}_N^{\boxtimes n}$ of the confusability graph $\mathtt{G}_N$\footnote{Given two graphs $\mathtt{G} = (V(\mathtt{G}), E(\mathtt{G}))$ and $\mathtt{H} = (V(\mathtt{H}), E(\mathtt{H}))$, the strong product $\mathtt{G} \boxtimes \mathtt{H}$ is a graph with vertex set $V(\mathtt{G} \boxtimes \mathtt{H}) = V(\mathtt{G}) \times V(\mathtt{H})$, and two \textit{distinct} vertices $(g, h)$ and $(g', h')$ are adjacent in $\mathtt{G} \boxtimes \mathtt{H}$ if and only if $(g,h) \sim (g',h')\iff\big(g = g' \vee g \sim_{\mathtt{G}} g'\big) \wedge \big(h = h' \vee h \sim_{\mathtt{H}} h'\big)$.}. The asymptotic zero-error capacity of $N$ is defined as
\begin{align}
C_\circ(N)=\lim_{n\to\infty}\frac{1}{n}c_\circ\!\left(N^{\otimes n}\right).   
\end{align}
\noindent Strikingly, the zero-error capacity of a noisy classical channel can exhibit superadditivity; that is, there exist channels for which $c_\circ(N^{\otimes n}) > n\,c_\circ(N)$. A classic example was given by Shannon \cite{Shannon1956}, who considered a channel $N_5$ whose confusability graph is the cycle graph $\mathtt{C}_5$. In this case, $\alpha(\mathtt{C}_5)=2$, whereas $ \alpha(\mathtt{C}_5^{\boxtimes 2})=5$. Subsequently, in another seminal work, Lovász \cite{Lovasz1979} proved that $ C_\circ(N_5)=\log_2 \sqrt{5}$. More generally, he showed that for any channel $N$ with confusability graph $\mathtt{G}_N$, the asymptotic zero-error capacity satisfies $C_\circ(N)\le \log_2 \vartheta(\mathtt{G}_N)$, where $\vartheta(\mathtt{G}_N)$ denotes the Lovász theta number of $\mathtt{G}_N$.

The potential advantages of quantum resources in communication tasks have been widely studied \cite{Gisin2002,Gisin2007,Buhrman2010,khatri2024,Usenko2025}. More recently, their role in zero-error communication has also been investigated \cite{Cubitt2010,Cubitt2011,Prevedel2011,Yadavalli2022,Alimuddin2023,Agarwal2026}. In particular, the zero-error capacity of a classical channel can be defined in the presence of pre-shared correlations between the sender and the receiver. Three classes of correlations are typically considered: (i) classical correlations, also known as shared randomness (SR); (ii) shared entanglement (SE), where the parties share an entangled quantum state \cite{Horodecki2009}; and (iii) non-signaling (NS) correlations, which do not enable communication by themselves but can be more general than quantum correlations \cite{Brunner2014}. The single-shot assisted zero-error capacity $c_{\Omega}(N)$ [accordingly the asymptotic capacity $C_{\Omega}(N)$] of a channel $N$ is defined with respect to the class of correlations $\Omega\in\{\circ, SR,SE,NS\}$; where `$\circ$' denotes no assistance. The following observations as obtained in \cite{Cubitt2010,Cubitt2011} are noteworthy.

\begin{observation}\label{obs1}
(i) For every channel $N$, $c_\circ(N)=c_{SR}(N)$;
(ii) There exist channels $N$ such that $c_\circ(N)>0$ and $c_{SE}(N)>c_\circ(N)$;
(iii) There exist channels $N$ such that $c_\circ(N)=0$ and $c_{NS}(N)>0$.
\end{observation}
\noindent Notably, shared entanglement can increase the zero-error capacity of a noisy classical channel only when the unassisted capacity is nonzero; otherwise such an increase would violate causality \cite{Cubitt2010}.

One may also consider a different form of assistance, in which a classical channel $N$ is supplemented by an additional noiseless channel. The assisting channel may itself be classical or quantum. Let $\mathrm{id}^c_k$ denote a noiseless classical channel of alphabet size $k$, i.e., $\mathrm{id}^c_k : X \to X$, where $|X|=k$ and the channel transmits the input symbol perfectly to the output. Similarly, let $\mathrm{id}^q_k$ denote a noiseless quantum channel acting on a $k$-dimensional Hilbert space, $\mathrm{id}^q_k : \mathcal{L}(\mathbb{C}^k) \to \mathcal{L}(\mathbb{C}^k)$, defined as the completely positive trace-preserving (CPTP) map $\mathrm{id}^q_k(\rho)=\rho,~ \forall\, \rho \in \mathcal{D}(\mathbb{C}^k)$ \cite{Kraus1983}. In this setting, the following observations made in \cite{Ambuj2026} are noteworthy.

\begin{observation}\label{obs2}
(i) For every classical channel $N$ and every $k\in\mathbb{N}$, $c_\circ(N\otimes \mathrm{id}^c_k)= c_\circ(N)+c_\circ(\mathrm{id}^c_k)=c_\circ(N)+k$, which follows from the properties of the strong product of graphs; (ii) There exist classical channels $N$ and integers $k$ such that $c_\circ(N\otimes \mathrm{id}^q_k) > c_\circ(N)+k$.
\end{observation}

\noindent Notably, the zero-error capacity enhancement observed in \cite{Ambuj2026} is qualitatively different from the capacity enhancement discussed in \cite{Cubitt2010,Cubitt2011,Prevedel2011,Yadavalli2022,Agarwal2026}. 

\section{Results with classical channels}\label{sec3}
\subsection{Setup and Preliminaries}
Given a classical channel $N\equiv\{P(y|x)~|~\sum_yP(y|x)=1\}$ from input alphabet $X$ to output alphabet $Y$, we can define the output-range of an input and input-domain of an output as follows.
\begin{definition}[Output-Range \& Input-Domain]\label{def4}
For a channel $N : X \to Y$, the output-range of $x \in X$ is the set of outputs that occur with nonzero probability, i.e. $\Gamma_x := \{ y \in Y \mid P(y \mid x) > 0 \}$; and the input-domain of $y \in Y$ is the set of inputs that can produce that output with nonzero probability, i.e. $\Omega_y := \{ x \in X \mid P(y \mid x) > 0 \}$.
\end{definition}

\noindent For instance, consider a channel $N:X\to Y$, with input alphabet $X=\{a,b,c\}$ and output alphabet $Y=\{1,2\}$, defined by 
\begin{align}
N:=\left(
\begin{array}{c||c|c|c}
Y\backslash X         & x=a & x=b        & x=c \\
\hline\hline
y=1      & 1   & 1/3 & 0   \\\hline
y=2      & 0   & 2/3 & 1
\end{array}
\right). \label{Ex-Sup}
\end{align}
\noindent The input-range and output-domain are depicted in a bipartite hypergraph in Fig.~\ref{fig1}.

\begin{figure}[b!]
\centering
\includegraphics[width=0.5\textwidth]{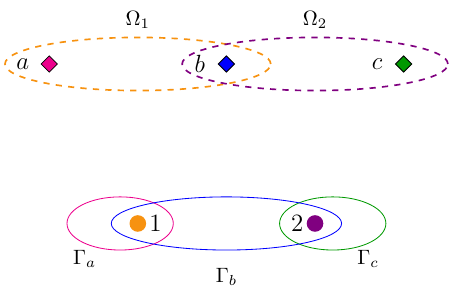}
\caption{Bipartite hypergraph representing the input domain and output range of the channel $N:\{a,b,c\}\to\{1,2\}$ defined in Eq.~(\ref{Ex-Sup}).}\label{fig1}
\end{figure}

\begin{definition}[XY-Equivalent Channel]\label{def5}
A channel $N : X \to Y$ is called XY-equivalent if $X = Y$, i.e., the input and output alphabets coincide.   
\end{definition}

\noindent For such a channel the output-range and input-domain can be depicted in a digraph instead of bipartite hypergraph. For instance, consider the XY-equivalent channel $N:X\to X$, with $X=\{1,2,3\}$, defined by 
\begin{align}
N:=\left(
\begin{array}{c||c|c|c}
Y\backslash X         & x=1 & x=2 & x=3 \\
\hline\hline
y=1      & 0   & 1/3 & 0   \\\hline
y=2      & 1/4   & 1/3 & 0  \\\hline
y=3      & 3/4   & 1/3 & 1  
\end{array}
\right). \label{digraph}
\end{align} 

\begin{figure}[t]
\centering
\includegraphics[width=0.35\textwidth]{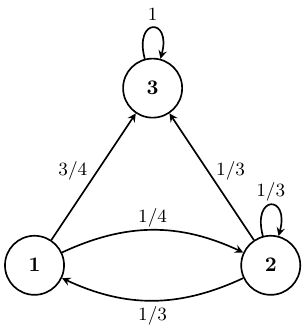}
\caption{Weighted digraph corresponding to the XY-equivalent channel of Eq.~(\ref{digraph}). A directed edge starting from a vertex $x$ ends in the output-range $\Gamma_x$ of that input. A directed edge ending at a vertex $x$ originates from the input-domain $\Omega_x$ of that output. Each edge is weighted by a transition probability in $(0, 1]$, denoting the probability of transmitting that particular input to that particular output.}
\label{fig2}
\end{figure}
\noindent The weighted digraph describing this channel is shown in Fig.~\ref{fig2}. 

\begin{definition}[Fully-Corrupted Input]\label{def6}
For an XY-equivalent channel $N : X \to X$, an input $x \in X$ is called fully-corrupted if $P(x \mid x) = 0$, otherwise it is called not-fully-corrupted. 
\end{definition}

\noindent Note that a not-fully-corrupted input $x \in X$ does not imply $P(x \mid x) = 1$; it only implies $P(x \mid x) > 0$. 

\begin{definition}[Symmetric Not-Fully-Corrupted Channel]\label{def:SNFC}
An XY-equivalent channel $\mathbb{N}\equiv\{N_{x'x}:=P(x'|x)~\mid~x,x'\in X\}$ is called symmetric not-fully-corrupted (SNFC) if 
\begin{equation*}
N_{x'x} = 0 \iff N_{xx'} = 0, \quad \forall\, x, x' \in X,
\end{equation*}
and additionally $P(x \mid x) > 0$ for all $x \in X$.
\end{definition}
   
\noindent For an SNFC channel the digraph representation of Fig.~\ref{fig2} can be further simplified.

\begin{definition}[Support Graph of an SNFC Channel]\label{def:support-graph-SNFC}
For an SNFC channel $N : X \to X$, the support graph is an undirected graph $\mathtt{S}_N = (X, E)$, where $\{x, x'\} \in E \iff N_{x'x} > 0$.
\end{definition}

\noindent Since $P(x \mid x) > 0$ for all $x \in X$, every vertex of $\mathtt{S}_N$ has a self-loop. However, for simplicity of representation, we omit the self-loops with the understanding that they are always present. The support graph $\mathtt{S}_N$ corresponds to the adjacency matrix $\mathbb{S}_N = (S_{x'x})$ given by
\begin{align}
S_{x'x} :=
\begin{cases}
1 & \text{if } N_{x'x} > 0, \\
0 & \text{otherwise}.
\end{cases}
\end{align}
\noindent The adjacency matrix of an SNFC channel is a symmetric matrix. Given a channel, its adjacency matrix is unique; however, the same adjacency matrix can arise from different channels. The confusability graph and the support graph of an SNFC channel are related to each other as follows.

\begin{proposition}\label{prop1}
For an SNFC channel $N : X \to X$ with support graph adjacency matrix $\mathbb{S}_N = (S_{xx'})$, the adjacency matrix $\mathbb{G}_N = (G_{xx'})$ of the confusability graph $\mathtt{G}_N$ is given by the non-zero pattern of $\mathbb{S}_N^2$, i.e.
\begin{equation}
G_{xx'} = 
\begin{cases}
1 & \emph{if } (\mathbb{S}_N^2)_{xx'} \neq 0, \\
0 & \emph{otherwise,}
\end{cases}
\end{equation}
or equivalently, $G_{xx'} \neq 0 \iff \sum_{\tilde{x} \in X} S_{x\tilde{x}}\, S_{\tilde{x}x'} > 0$.
\end{proposition}

\begin{proof}
Since $N$ is SNFC, the adjacency matrix $\mathbb{S}_N$ is symmetric, $\mathbb{S}_N^\top = \mathbb{S}_N$, with all diagonal entries equal to $1$ (self-loops). The $(x, x')$ entry of $\mathbb{S}_N^2$ is $(\mathbb{S}_N^2)_{xx'} = \sum_{\tilde{x} \in X} S_{x\tilde{x}}\, S_{\tilde{x}x'}$. Since $S_{x\tilde{x}}, S_{\tilde{x}x'} \in \{0,1\}$, each term in the sum is binary, so $(\mathbb{S}_N^2)_{xx'}$ counts the number of common neighbours of $x$ and $x'$ in 
$\mathtt{S}_N$. We consider two cases.

\noindent\textbf{Case 1:} $(\mathbb{S}_N^2)_{xx'} \neq 0$. Then there exists $\tilde{x} \in X$ such that $S_{x\tilde{x}} = 1$ and $S_{\tilde{x}x'} = 1$, i.e. $N_{\tilde{x}x} > 0$ and $N_{\tilde{x}x'} > 0$. Hence output $\tilde{x}$ is reachable from both inputs $x$ and $x'$, so $x$ and $x'$ are confusable, giving $G_{xx'} = 1$.

\noindent\textbf{Case 2:} $(\mathbb{S}_N^2)_{xx'} = 0$. Then for all $\tilde{x} \in X$, either $S_{x\tilde{x}} = 0$ or $S_{\tilde{x}x'} = 0$, meaning no output is reachable from both $x$ and $x'$ simultaneously. Hence $x$ and $x'$ are not confusable, giving $G_{xx'} = 0$.

\noindent Combining both cases: $G_{xx'} = 1 \iff (\mathbb{S}_N^2)_{xx'} \neq 0\iff \exists\, \tilde{x} \in X : N_{\tilde{x}x} > 0 \text{ and } N_{\tilde{x}x'} > 0$. Finally, since $P(x \mid x) > 0$ for all $x \in X$ by the SNFC condition, we have $S_{xx} = 1$ for all $x$, so $(\mathbb{S}_N^2)_{xx} \geq S_{xx} \cdot S_{xx} = 1$ for all $x \in X$, confirming that every vertex carries a self-loop in $\mathtt{G}_N$.
\end{proof}

\begin{figure}[t!]
\centering
\includegraphics[width=0.5\textwidth]{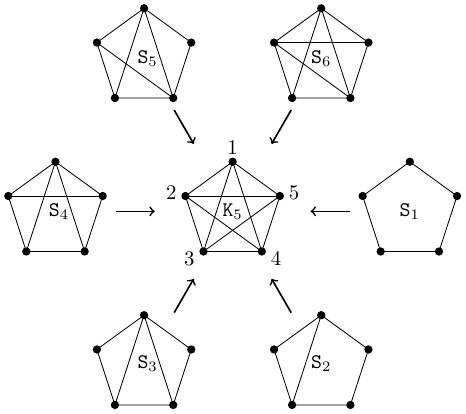}
\caption{Support graphs $\mathtt{S}_i$ (self-loops omitted) arranged clockwise with increasing number of chords ($\mathtt{S}_4$ and $\mathtt{S}_5$ have same number of chords), all leading to the same confusability graph $\mathtt{G}_{N_i}=\mathtt{K}_5$.}\label{fig3}
\end{figure}

\noindent Notably, a support graph corresponds to a unique confusability graph, but the same confusability graph can arise from different support graphs, and hence from different channels. For instance, consider the SNFC channels $N_i : X \to X$ for $i \in \{1, \cdots, 6\}$ with $X := \{1,\cdots, 5\}$, whose adjacency matrices (denoted $\mathbb{S}_{N_i}=\mathbb{S}_i$) are given by
\begin{align}
{\fontsize{8}{9}\selectfont
\begin{aligned}
\mathbb{S}_1 = 
\begin{pmatrix}
        1 & 1 & 0 & 0 & 1 \\
        1 & 1 & 1 & 0 & 0 \\
        0 & 1 & 1 & 1 & 0 \\
        0 & 0 & 1 & 1 & 1 \\
        1 & 0 & 0 & 1 & 1
\end{pmatrix};~~
\mathbb{S}_2 = 
\begin{pmatrix}
        1 & 1 & 1 & 0 & 1 \\
        1 & 1 & 1 & 0 & 0 \\
        1 & 1 & 1 & 1 & 0 \\
        0 & 0 & 1 & 1 & 1 \\
        1 & 0 & 0 & 1 & 1
\end{pmatrix};~~
\mathbb{S}_3= 
\begin{pmatrix}
        1 & 1 & 1 & 1 & 1 \\
        1 & 1 & 1 & 0 & 0 \\
        1 & 1 & 1 & 1 & 0 \\
        1 & 0 & 1 & 1 & 1 \\
        1 & 0 & 0 & 1 & 1
\end{pmatrix};\\
\mathbb{S}_4 = 
\begin{pmatrix}
        1 & 1 & 1 & 1 & 1 \\
        1 & 1 & 1 & 0 & 1 \\
        1 & 1 & 1 & 1 & 0 \\
        1 & 0 & 1 & 1 & 1 \\
        1 & 1 & 0 & 1 & 1
\end{pmatrix};~~
\mathbb{S}_5 = 
\begin{pmatrix}
        1 & 1 & 1 & 1 & 1 \\
        1 & 1 & 1 & 1 & 0 \\
        1 & 1 & 1 & 1 & 0 \\
        1 & 1 & 1 & 1 & 1 \\
        1 & 0 & 0 & 1 & 1
\end{pmatrix};~~
\mathbb{S}_6= 
\begin{pmatrix}
        1 & 1 & 1 & 1 & 1 \\
        1 & 1 & 1 & 1 & 1 \\
        1 & 1 & 1 & 1 & 0 \\
        1 & 1 & 1 & 1 & 1 \\
        1 & 1 & 0 & 1 & 1
\end{pmatrix}.
\end{aligned}}\label{penta}
\end{align}

\noindent The confusability graph for all these channels is the complete graph $\mathtt{K}_5$ (see Fig.~\ref{fig3}).

\subsection{Conclusive Identification Task}
\begin{question}\label{qu1}
Consider an SNFC channel $N : X \to X$ with $X := \{1, \ldots, n\}$, having confusability graph $\mathtt{G}_N = \mathtt{K}_n$. Such a channel exhibits the following properties:
\begin{itemize}[itemsep=-.15cm, parsep=.15cm, topsep=2pt, leftmargin=1.0cm]
\item[\emph{P1)}] $c_\circ(N) = C_\circ(N) = 0$ \cite{Shannon1956};
\item[\emph{P2)}] $c_{SE}(N) = C_{SE}(N) = 0$ \cite{Cubitt2010};
\item[\emph{P3)}] $c_\circ(N \otimes \mathrm{id}^c_d) = C_\circ(N \otimes \mathrm{id}^c_d) = \log_2 d$ \text{ bits};
\item[\emph{P4)}] $c_\circ(N \otimes \mathrm{id}^q_d) = C_\circ(N \otimes \mathrm{id}^q_d) = \log_2 d$ \text{ bits} \cite{Ambuj2026}.
\end{itemize}
Does this mean the channel $N$ cannot exhibit any superactivation or superadditivity behavior in some communication task?
\end{question}

\noindent We address this question in affirmative by introducing the following communication task. 

\begin{task}[Conclusive Identification]\label{task-CI}
Given XY-equivalent channel $N : X \to X$ from a sender (Alice) to a receiver (Bob), in conclusive identification task, Bob aims to identify the transmitted message conclusively whenever it arrives uncorrupted. When a message is corrupted, Bob is allowed to respond inconclusively, meaning he is not required to identify the message.
\end{task}

\noindent The utility of a channel in conclusive identification task can be quantified through an index defined as follows. 

\begin{definition}[Conclusive Identification Index]\label{def9}
For an XY-equivalent channel $N : X \to X$, the unassisted single-shot conclusive identification index, denoted $\mathrm{ci}_\circ(N)$, is the maximum number of inputs that the receiver can conclusively identify without error. 
\end{definition}

\noindent At this point it is important to note that the single-shot zero-error capacity $c_\circ$ and the single-shot conclusive identification index $\mathrm{ci}_\circ$ capture different communication aspects of an XY-equivalent channel, and the two quantities are in general incomparable. The key structural difference between the two quantities can be summarized as follows:
\begin{itemize}[itemsep=-0.15cm, parsep=0.15cm, topsep=2pt, leftmargin=0.5cm]
\item The single-shot zero-error capacity $c_\circ$ measures the maximum number of messages that can be transmitted with \textit{zero} probability of error. It depends only on the global structure of the confusability graph $\mathtt{G}_N$ of the channel, and in particular on its independence number $\alpha(\mathtt{G}_N)$.
\item On the other hand, $\mathrm{ci}_\circ$ measures the maximum number of inputs that the receiver can identify \textit{conclusively} — not all inputs need to be decodable, and the receiver is permitted to respond inconclusively when the received output is ambiguous. It depends on the \textit{local} structure of the support graph $\mathtt{S}_N$, specifically on the diagonal entries $P(x \mid x)$ of the channel matrix.
\end{itemize}

\noindent The conclusive identification task introduced here is conceptually related to, but distinct from, the classical identification framework of Ahlswede and Dueck~\cite{Ahlswede1989-1} (see \cite{Ahlswede1989-2,Ahlswede2002,Winter2013,Colomer2025} for related works and generalization of the problem in quantum setup). In the Ahlswede-Dueck setting, the receiver's goal is to answer the binary question ``was message $i$ sent?'',  with two types of errors -- missed identification ($\lambda_1$) and false identification ($\lambda_2$) -- both required to vanish asymptotically. Our task differs in three fundamental ways. First, we operate in the \textit{zero-error} regime: a conclusive response from Bob must always be correct, corresponding to $\lambda_2 = 0$ exactly. Second, Bob is explicitly permitted to respond \textit{inconclusively} when the received output is ambiguous, effectively allowing $\lambda_1 \leq 1$. This is analogous to unambiguous discrimination in quantum information theory \cite{Chefles2000,Barnett2009,Bae2015}. Third, we operate in the single-shot regime, measuring the number of conclusively identifiable inputs rather than an asymptotic rate.

\subsection{Conclusive Identification with Forward Classical Assistance}

It is evident that for any SNFC channel $N$ whose support graph is one of those in Eq.~(\ref{penta}), has conclusive identification index $\mathrm{ci}_\circ(N) = 0$. A natural follow-up question is whether the noisy channel $N$ can contribute genuine utility in the conclusive identification task when assisted by a perfect classical forward channel $\mathrm{id}^c_d$ from Alice to Bob. Specifically, one can ask the following.
\begin{question}\label{qu2}
For an SNFC channel $N : X \to X$ with $|X| = 5$ and $\mathrm{ci}_\circ(N) = 0$, does there exist $d < 5$ such that $\mathrm{ci}_\circ(N \otimes \mathrm{id}^c_d) = 5?$
\end{question}

\noindent An affirmative answer to Question~\ref{qu2} would establish an intriguing \emph{superactivation} phenomenon in the conclusive identification task: the noisy channel $N$, which is entirely useless for conclusive identification on its own (i.e. $\mathrm{ci}_\circ(N) = 0$), gets \emph{activated} by the assistance of a perfect classical channel $\mathrm{id}^c_d$ of size $d < 5$, yielding
\begin{equation*}
\mathrm{ci}_\circ(N \otimes \mathrm{id}^c_d) = 5 > \mathrm{ci}_\circ(N) + \mathrm{ci}_\circ(\mathrm{id}^c_d)  = 0 + d = d,
\end{equation*}
which strictly exceeds the sum of the individual indices. In other words, the two channels together can conclusively identify all $5$ inputs, even though neither channel alone can achieve this.

\subsection{Superactivation of Conclusive Identification Index}
We start by showing the superactivation of conclusive identification index with an explicit example. 

\begin{lemma}\label{lemma1}
For an SNFC channel $N : X \to X$ with $|X| = 5$ whose support graph is $\mathtt{S}_1$ of Eq.~(\ref{penta}), the receiver can conclusively identify all $5$ inputs when assisted with a perfect classical channel $\mathrm{id}^c_3$, i.e. $\mathrm{ci}_\circ(N \otimes \mathrm{id}^c_3) = 5$.
\end{lemma}

\begin{proof}
Let the channel $N$ has input (and output) $X=\{1,\cdots,5\}$, the inputs (outputs) of $\mathrm{id}^c_3$ channel are $\{red (R), green (G), blue (B)\}$. The combined channel $N \otimes \mathrm{id}^c_3$ has input and alphabet $X \times \{R,G,B\}$. Alice partitions $X = \{1,2,3,4,5\}$ into three groups according to the map $\mathcal{P}: X \to \{R,G,B\}$:
\begin{equation*}
\mathcal{P}(2)= \mathcal{P}(4)= B,\quad \mathcal{P}(1) = R, \quad \mathcal{P}(3) = \mathcal{P}(5) = G.
\end{equation*}
Upon receiving the pair $(x, \mathcal{P})\in X\times\{R,G,B\}$, Bob knows which group the input belongs to, and uses the output $x$ of the noisy channel to identify the input conclusively within that group. We verify each group separately in Table~\ref{tab1}.

\begin{table}[t]
\centering
\begin{tabular}{c|c|c|c}
\hline
~Input of $N$~ & ~Output of $N$~ &~ Message via $\mathrm{id}^c_3$~ & ~Identified?~ \\
\hline\hline
\multirow{3}{*}{$1$} 
& $5$ & R & $\checkmark$ \\
& $1$ & R & $\checkmark$ \\
& $2$ & R & $\checkmark$ \\
\hline
\multirow{3}{*}{$2$} 
& $1$ & B & $\checkmark$ \\
& $2$ & B & $\checkmark$ \\
& $3$ & B & $\times$ \\
\hline
\multirow{3}{*}{$3$} 
& $2$ & G & $\checkmark$ \\
& $3$ & G & $\checkmark$ \\
& $4$ & G & $\times$ \\
\hline
\multirow{3}{*}{$4$} 
& $3$ & B & $\times$ \\
& $4$ & B & $\checkmark$ \\
& $5$ & B & $\checkmark$ \\
\hline
\multirow{3}{*}{$5$} 
& $4$ & G & $\times$ \\
& $5$ & G & $\checkmark$ \\
& $1$ & G & $\checkmark$ \\
\hline
\end{tabular}
\caption{The conclusive identification scheme for the SNFC channel $N$ (having support graph $\mathtt{S}_{N}=\mathtt{S}_1$) when assisted by a noiseless classical channel $\mathrm{id}^c_3$. }
\label{tab1}
\end{table}
\noindent We have $\mathrm{ci}_\circ(N) = 0$ and $\mathrm{ci}_\circ(\mathrm{id}^c_3) = 3$, whereas the combined channel $N \otimes \mathrm{id}^c_3$ allows Bob to conclusively identify all $5$ inputs, giving $\mathrm{ci}_\circ(N \otimes \mathrm{id}^c_3) = 5$, 
which establishes the superadditivity of $\mathrm{ci}_\circ$. This completes the proof. 
\end{proof}

\begin{remark}\label{remark1}
The proof of Lemma~\ref{lemma1} admits a natural graph-theoretic interpretation. Specifically, the partition of $X=\{1,\cdots,5\}$ into three groups corresponds precisely to a proper $3$-coloring of the support graph $\mathtt{S}_1$, as depicted in Fig.~\ref{fig4}. Each color class defines a group of inputs assigned the same assisted message, and the conclusive identifiability of every input within its group follows from the fact that no two adjacent vertices in $\mathtt{S}_1$ share the same color, ensuring that their output-ranges are sufficiently separated for Bob to respond conclusively.
\begin{figure}[b!]
\centering
\includegraphics[width=0.55\textwidth]{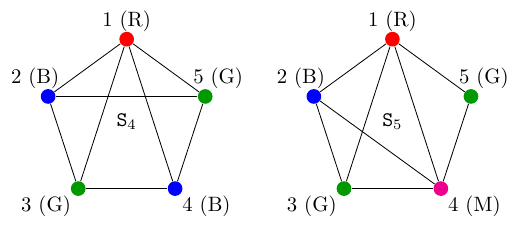}
\caption{(Left) Coloring of the support graph $\mathtt{S}_4$. The same coloring holds for $\mathtt{S}_3$, $\mathtt{S}_2$, and $\mathtt{S}_1$: each of these graphs has fewer edges (chords) than $\mathtt{S}_4$ and thus requires no additional colors. Although $\mathtt{S}_5$ has the same number of edges as $\mathtt{S}_4$, it requires one extra color (right).  Unlike $\mathtt{S}_4$, its degree sequence is more asymmetric, making a valid coloring with the same number of colors impossible.}
\label{fig4}
\end{figure}
\end{remark}

\begin{remark}\label{remark2}
Lemma~\ref{lemma1} generalizes to all SNFC channels $N : X \to X$ with $|X| = 5$ and $\mathtt{S}_N \in \{\mathtt{S}_i\}_{i=1}^{6}$ as follows.
\begin{itemize}[itemsep=-0.15cm, parsep=0.15cm, topsep=2pt, leftmargin=1.2cm]
\item[\textbf{Case 1.}] If $\mathtt{S}_N \in\{\mathtt{S}_i\}_{i=1}^{4}$, then $\mathrm{ci}_\circ(N \otimes \mathrm{id}^c_3) = 5$. The partition $\mathcal{P} : X \to \{\mathrm{R}, \mathrm{G}, \mathrm{B}\}$ is given by $\mathcal{P}(1)=\mathrm{R},~\mathcal{P}(2) = \mathcal{P}(4) = \mathrm{B},~\mathcal{P}(3) = \mathcal{P}(5) = \mathrm{G}$. The corresponding identification scheme is a modification of Table~\ref{tab1}, adapted according to $\mathtt{S}_i$ \emph{(see Appendix {\bf A1})}.

\item[\textbf{Case 2.}] If $\mathtt{S}_N \in\{\mathtt{S}_i\}_{i=5}^{6}$, then $\mathrm{ci}_\circ(N \otimes \mathrm{id}^c_4) = 5$. The partition $\mathcal{P} : X \to \{\mathrm{R}, \mathrm{G}, \mathrm{B}, \mathrm{M}\}$ is given by $\mathcal{P}(1) = \mathrm{R},~  \mathcal{P}(2) = \mathrm{B},~\mathcal{P}(3) = \mathcal{P}(5) = \mathrm{G},~\mathcal{P}(4) = \mathrm{M}$, where $\mathrm{M}$ denotes a fourth color (magenta) [see Fig.~\ref{fig4}].
\end{itemize}
\end{remark}

\noindent At this point we would like to emphasize that any SNFC channel with support graph $\mathtt{S}_{N_i}$, $i \in \{1, \ldots, 6\}$, has confusability graph $K_5$, and thus satisfies properties (P1)--(P4) listed in Question~\ref{qu1}. In particular, from the perspective of the confusability graph, these channels are uninteresting as they cannot exhibit any superadditivity or superactivation behavior in the zero-error communication framework. However, Lemma~\ref{lemma1} demonstrates that analyzing the same channels through the lens of their support graph reveals an intriguing superactivation phenomenon in the conclusive identification task. Furthermore, channels that are indistinguishable from one another under the confusability graph framework can exhibit fundamentally different operational behavior when their support graph structure is taken into account. This highlights a key limitation of the confusability graph approach and motivates the support graph as a strictly finer characterization of channels' utility.

\begin{observation}\label{obs3}
The channels corresponding to support graphs $\mathtt{S}_4$ and $\mathtt{S}_5$ are worth examining closely. Both graphs are obtained by adding exactly three chords to the cycle graph $\mathtt{C}_5$. However, they differ in the amount of classical assistance required for conclusive identification: $\mathrm{id}^c_3$ suffices for $\mathtt{S}_4$, whereas $\mathrm{id}^c_4$ is required for $\mathtt{S}_5$. This difference is reflected in their degree sequences. Let $i[\ell]$ denotes the degree of vertex $i$. The degree configurations of $\mathtt{S}_4,~\mathtt{S}_5$ are
\begin{subequations}
\begin{align}
\mathtt{S}_4:~~~\big\{1[4],\ 2[3],\ 3[3],\ 4[3],\ 5[3]\big\},\\  
\mathtt{S}_5:~~~\big\{1[4],\ 2[3],\ 3[3],\ 4[4],\ 5[2]\big\},
\end{align}    
\end{subequations}
making $\mathtt{S}_5$ more asymmetric than $\mathtt{S}_4$ (see Fig.~\ref{fig4}). It is this asymmetry that necessitates a finer partition (coloring), and hence a larger classical assistance.
\end{observation}

\noindent This observation naturally raises the following question: what is the minimum perfect classical assistance required to activate an SNFC channel for conclusive identification? We answer this question in the following theorem.

\begin{theorem}\label{theo1}
For an SNFC channel $N : X \to X$ with $|X| = n$, $c_\circ(N)=0$, $\mathrm{ci}_\circ(N) = 0$, and support graph $\mathtt{S}_N$, a perfect classical channel $\mathrm{id}^c_\beta$ is necessary and sufficient to achieve $\mathrm{ci}_\circ(N \otimes \mathrm{id}^c_\beta) = n$, with $\beta$ being the chromatic number\footnote{A proper graph coloring is not defined for graphs with self-loops, since no valid coloring exists in that case. While the support graph $\mathtt{S}_N$ includes a self-loop at each vertex $x$, here we consider the associated simple graph obtained by deleting all self-loops. The chromatic number $\chi(\mathtt{S}_N)$ is thus well-defined.} $\chi(\mathtt{S}_N)$ of the support graph $\mathtt{S}_N$. 
\end{theorem}

\begin{proof}
\textbf{Sufficiency.} Assign to each input $x \in X$ a color $\mathcal{P}(x) \in \{1, \ldots, \beta\}$ according to a proper $\beta$-coloring of $\mathtt{S}_N$, and transmit $\mathcal{P}(x)$ to Bob via $\mathrm{id}^c_\beta$. Upon receiving $(y, \mathcal{P}(x))$, Bob restricts attention to the color class $\mathcal{P}^{-1}(\mathcal{P}(x))$. Since $\mathcal{P}$ is a proper coloring, no two vertices in the same color class are adjacent in $\mathtt{S}_N$, meaning their output-ranges do not overlap via shared neighbours. In particular, $x$ has a private output 
\begin{align*}
y \in \Gamma_x \setminus \cup_{x' \in \mathcal{P}^{-1}(\mathcal{P}(x)),\, x' \neq x} 
\Gamma_{x'}   
\end{align*}
within its color class, allowing Bob to conclusively identify $x$. Since this holds for all $x \in X$, we have $\mathrm{ci}_\circ(N \otimes \mathrm{id}^c_\beta) = n$.

\noindent\textbf{Necessity.} Suppose $\mathrm{id}^c_k$ with $k < \beta$ suffices, i.e. $\mathrm{ci}_\circ(N \otimes \mathrm{id}^c_k)= n$. Then the partition of $X$ induced by Alice's encoding strategy defines a proper $k$-coloring of $\mathtt{S}_N$: two inputs assigned the same group must have disjoint output-ranges within that group, which requires them to be non-adjacent in $\mathtt{S}_N$. But this contradicts $\chi(\mathtt{S}_N) = \beta > k$. Hence $k \geq \beta$, and $\mathrm{id}^c_\beta$ is necessary. This completes the proof.
\end{proof}

\noindent Theorem~\ref{theo1} reveals an intriguing connection between the conclusive identification task and the graph coloring problem.

\subsection{Arbitrarily Large Superactivation}
As pointed out in Remark~\ref{remark2}, the amount of classical assistance required to achieve conclusive identification varies depending on the support graph of the channel. Specifically, while $\mathrm{id}^c_4$ is required for channels with support graphs in $\{\mathtt{S}_i~\mid~i=5,6\}$, an assistance of $\mathrm{id}^c_3$ suffices for channels with support graphs in $\{\mathtt{S}_i~\mid~i=1,\cdots,4\}$, thereby exhibiting a higher degree of superactivation. This naturally raises the question of whether the gap $\Delta(N):=  \mathrm{ci}_\circ(N \otimes \mathrm{id}^c_\beta) -  \mathrm{ci}_\circ(\mathrm{id}^c_\beta) = n - \beta$ can be made arbitrarily large. Our next result answers this question affirmatively.  

\begin{theorem}\label{theo2}
For an SNFC channel $N : X \to X$ with $|X| = n$ and support graph $\mathtt{S}_N$ being the cycle graph $\mathtt{C}_n$ (with each vertices having self-loop), the following hold:
\begin{itemize}[itemsep=-0.15cm, parsep=0.15cm, topsep=2pt, leftmargin=0.8cm]
\item[\emph{(i)}] $\mathrm{ci}_\circ(N) = 0$;
\item[\emph{(ii)}] $\mathrm{ci}_\circ(N \otimes \mathrm{id}^c_\beta)=n$, where $\beta = \chi(\mathtt{C}_n)=2$ for even $n$ and $\beta=3$ for odd $n$;
\item[\emph{(iii)}] The superactivation gap $\Delta(N)=n-\beta$,which grows arbitrarily large as $n \to \infty$.
\end{itemize}
\end{theorem}

\begin{proof}
(i) Since $\mathtt{S}_N =\mathtt{C}_n$, the output-range of each input $x \in X$ is
\begin{equation*}
\Gamma_x = \{x-1,\ x,\ x+1\} \pmod{n},
\end{equation*}
where the self-loop contributes $x$ itself. For any output $y \in X$, the input-domain is
\begin{equation*}
\Omega_y = \{y-1,\ y,\ y+1\} \pmod{n},
\end{equation*}
so every output is reachable from exactly three distinct inputs. Hence no output $y$ exclusively identifies any single input, i.e. $|\Omega_y| \geq 3 > 1$ for all $y \in X$, giving $\mathrm{ci}_\circ(N) = 0$.\\

\noindent (ii) By Theorem~\ref{theo1}, it suffices to exhibit a proper $\beta$-coloring of $\mathtt{C}_n$ such that within each color class, every input has a private output.

\noindent{\bf Case 1} ($n$ even, $\beta = 2$). Color the vertices with two colors:
\begin{equation*}
\mathcal{P}(x) \equiv x \pmod{2},\quad \mathcal{P}(x) \in \{1, 2\}.
\end{equation*}
This is a proper $2$-coloring of $\mathtt{C}_n$ since $n$ is even. The color classes are $A = \{1, 3, 5, \ldots, n-1\}$ and $B = \{2, 4, 6, \ldots, n\}$. For any $x \in A$, its output-range is $\Gamma_x = \{x-1, x, x+1\}$. Within class $A$, the nearest other vertex to $x$ is $x \pm 2$, with output-range $\Gamma_{x \pm 2} = \{x\pm1, x\pm2, x\pm3\}$. The private output of $x$ within $A$ is $\{x\} \subset \Gamma_x \setminus \bigcup_{x' \in A,\, x' \neq x} \Gamma_{x'}$, since $x \notin \Gamma_{x\pm2}$ as $|x - (x\pm2)| = 2 > 1$. The same argument applies to every vertex in $B$ by symmetry. Hence every input has a private output within its color class.

\noindent{\bf Case 2} ($n$ odd, $\beta = 3$). Color the vertices as:
\begin{equation*}
\mathcal{P}(x) \equiv x \pmod{3},\quad \mathcal{P}(x) \in \{1, 2, 3\}.
\end{equation*}
This is a proper $3$-coloring of $\mathtt{C}_n$ since no two adjacent vertices $x$ and $x+1$ share the same color modulo $3$. The color classes partition $X$ into three groups of sizes $\lceil n/3 \rceil$ or $\lfloor n/3\rfloor$. Within each color class, consecutive vertices are at distance $3$ in $C_n$, so their output-ranges $\Gamma_x = \{x-1, x, x+1\}$ and $\Gamma_{x+3} = \{x+2, x+3, x+4\}$ are disjoint. Hence the self-loop output $x \in \Gamma_x$ is private to $x$ within its color class, and every input is conclusively identifiable.\\

\noindent(iii) Since $\mathrm{ci}_\circ(\mathrm{id}^c_\beta) = \beta$, 
the superactivation gap is:
\begin{equation*}
\Delta(N)= \begin{cases} n - 2 & \text{if } n \text{ is even,} \\ n - 3 & \text{if } n \text{ is odd.} \end{cases}
\end{equation*}
In both cases the gap grows as $n - O(1)$, which is unbounded as $n \to \infty$. 
\end{proof}

\noindent On the single-shot zero-error capacity of the channels in Theorem~\ref{theo2}, we have the following result. 

\begin{proposition}\label{prop2}
For an SNFC channel $N : X \to X$ with $|X| = n$ and support graph $\mathtt{S}_N$ being the $\mathtt{C}_n$ (with self-loops), the single-shot 
zero-error capacity is
\begin{equation*}
    c_\circ(N) = 
    \begin{cases}
        0 & \text{if } n \leq 5, \\[4pt]
        \log_2 \left\lfloor \dfrac{n}{3} \right\rfloor 
        & \text{if } n \geq 6.
    \end{cases}
\end{equation*}
\end{proposition}

\begin{proof}
By Proposition~\ref{prop1}, the confusability graph of $N$ is $\mathtt{G}_N =\mathtt{C}_n^2$, where two vertices are adjacent if and only if their distance in $\mathtt{C}_n$ is at most $2$. Since $c_\circ(N) =\log_2 \alpha(\mathtt{G}_N)$, it suffices to compute $\alpha(\mathtt{C}_n^2)$.

\noindent{\bf Case 1} ($n \leq 5$). For $n \leq 5$, every pair of vertices is at distance at most $2$ in $\mathtt{C}_n$, so $\mathtt{C}_n^2 = K_n$, giving $\alpha(\mathtt{C}_n^2) = 1$ and hence $c_\circ(N) = 0$.

\noindent{\bf Case 2} ($n \geq 6$). A set $S \subseteq X$ is independent in $\mathtt{C}_n^2$ if and only if any two distinct vertices $x, x' \in S$ satisfy $|x - x'| \pmod{n} \geq 3$.

\noindent\textit{Upper bound.} Partition $X$ into consecutive blocks of size $3$:
\begin{equation*}
B_k = \{3k-2,\ 3k-1,\ 3k\}, \qquad k = 1, \ldots, \lfloor n/3 \rfloor,
\end{equation*}
together with at most $2$ remaining vertices if $3 \nmid n$. Within each block, every pair of vertices is at distance at most $2$ in $\mathtt{C}_n$, hence adjacent in $\mathtt{C}_n^2$. Any independent set therefore contains at most one vertex per block, giving $\alpha(\mathtt{C}_n^2) \leq \lfloor n/3 \rfloor$.

\noindent\textit{Lower bound.} The set $S = \{1,\ 4,\ 7,\ \ldots,\ 3\lfloor n/3 \rfloor - 2\}$ has $|S| = \lfloor n/3 \rfloor$ vertices. Any two consecutive elements of $S$ are at distance exactly $3$ in $\mathtt{C}_n$, hence non-adjacent in $\mathtt{C}_n^2$, so $S$ is an independent set and $\alpha(C_n^2) \geq \lfloor n/3 \rfloor$.

\noindent Combining both bounds gives $\alpha(C_n^2) = \lfloor n/3 \rfloor$, and hence $c_\circ(N) = \log_2 \left\lfloor \frac{n}{3} \right\rfloor$. This completes the proof.
\end{proof}

\noindent Theorem~\ref{theo2} establishes that
the superactivation gap can be made arbitrarily large for channels with support graph $\mathtt{C}_n$. However, it is worth noting that for $n \geq 6$, these channels already possess a nonzero zero-error capacity $c_\circ(N)=\log_2 \lfloor n/3 \rfloor > 0$, meaning they are not entirely useless from the perspective of the confusability graph framework. One may therefore ask whether the arbitrarily large superactivation gap can be achieved for channels satisfying both $c_\circ(N) = 0$ and $\mathrm{ci}_\circ(N) = 0$ simultaneously. Before addressing this question here we first recall a graph theoretic notion called diameter of a graph.

\begin{definition}[Distance and Diameter]\label{def:diameter}
Let $\mathtt{G} = (V, E)$ be a connected graph. The \emph{distance} between two vertices $u, v \in V$, denoted $d_{\mathtt{G}}(u, v)$, is the length of a shortest path between them, i.e.\ the minimum number of edges in any path from $u$ to $v$. The \emph{diameter} of $\mathtt{G}$ is the maximum distance over all pairs of vertices: $\mathrm{diam}(\mathtt{G}) :=\max_{u, v \in V} d_{\mathtt{G}}(u, v)$.
\end{definition}

\begin{proposition}\label{prop3}
For an SNFC channel $N : X \to X$ with support graph $\mathtt{S}_N$, the confusability graph satisfies $\mathtt{G}_N =\mathtt{K}_{|X|}$ if and only if $\mathrm{diam}(\mathtt{S}_N) \leq 2$.
\end{proposition}

\begin{proof}
By Proposition~\ref{prop1}, two inputs $x, x' \in X$ are confusable if and only if $\mathcal{N}[x] \cap \mathcal{N}[x'] \neq \emptyset$\footnote{For a graph $\mathtt{G} = (V, E)$ and a vertex $v \in V$, the open neighbourhood $\mathcal{N}(v):= \{u \in V \mid \{u, v\} \in E\}$ excludes $v$ itself, while the closed neighbourhood $\mathcal{N}[v]:= \mathcal{N}(v) \cup \{v\}$ includes it. However, since every vertex of the support graph $\mathtt{S}_N$ of an SNFC channel carries a self-loop by the condition $P(x \mid x) > 0$, the vertex $x$ is adjacent to itself, and thus $x \in \mathcal{N}(x)$. Consequently, $\mathcal{N}[x] = \mathcal{N}(x) \cup \{x\} = \mathcal{N}(x)$ for all $x \in X$ in our setting. We retain the notation $\mathcal{N}[x]$ throughout to emphasis the role of self-loops and to remain consistent with the standard graph theoretic convention.}. Since $N$ is SNFC, every vertex carries a self-loop, so $x \in \mathcal{N}[x]$ for all $x$. Therefore $\mathcal{N}[x] \cap \mathcal{N}[x'] \neq \emptyset$ if and only if there exists $\tilde{x} \in X$ with $\tilde{x} \in \mathcal{N}(x)$ and $\tilde{x} \in \mathcal{N}(x')$, i.e.\ $d(x,x') \leq 2$ in $\mathtt{S}_N$. Since this must hold for all pairs $x \neq x'$, the result follows.
\end{proof}

\begin{figure}[t]
\centering
\includegraphics[width=0.5\textwidth]{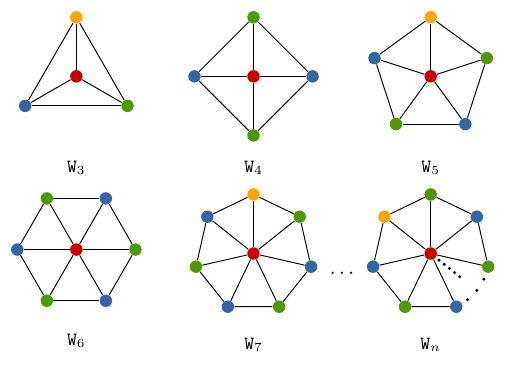}
\caption{Wheel Graph ($\mathtt{W}_n$). For even $n$ chromatic number $\chi(\mathtt{W}_n)=3$, whereas $\chi(\mathtt{W}_n)=4$ for odd $n$. $\mathrm{diam}(\mathtt{W}_3)=1$ and for $n>3$ $\mathrm{diam}(\mathtt{W}_n)=2$.}
\label{fig5}
\end{figure}

\begin{theorem}\label{theo3}
For an SNFC channel $N : X \to X$ with $|X| = n + 1$ and support graph $\mathtt{S}_N = \mathtt{W}_n$, the following hold:
\begin{itemize}[itemsep=-0.15cm, parsep=0.15cm, topsep=2pt, leftmargin=0.8cm]
\item[\emph{(i)}] $c_\circ(N) = 0$;
\item[\emph{(ii)}] $\mathrm{ci}_\circ(N) = 0$;
\item[\emph{(iii)}] $\mathrm{ci}_\circ(N \otimes \mathrm{id}^c_\beta) = n + 1$, where $\beta = \chi(\mathtt{W}_n) =3$  if $n$ is even and $\beta = 4$  if $n$ is odd;
\item[\emph{(iv)}] The superactivation gap $\Delta(N)=n+1-\beta$, which grows arbitrarily large as $n \to \infty$.
\end{itemize}
\end{theorem}

\begin{proof}
(i) The wheel graph $\mathtt{W}_n$ has $n+1$ vertices -- hub $h$ and rim $\{1, 2, \ldots, n\}$. Since $\mathrm{diam}(\mathtt{W}_n)\le2$ (see Fig.~\ref{fig5}), according to Proposition~\ref{prop3} we have confusability graph $\mathtt{G}_N=\mathtt{K}_{n+1}$, implying $c_\circ(N)=0$.
 
\noindent (ii) The output-ranges are given by
\begin{align*}
\Gamma_h&= \{h, 1, \ldots, n\}\qquad \text{and}\\
\Gamma_i&= \{h, i-1, i, i+1\} \pmod{n},   
\end{align*}
for each rim vertex $i$. For every output $y \in X$:
\begin{itemize}[itemsep=-0.15cm, parsep=0.15cm, topsep=2pt, leftmargin=0.4cm]
\item Output $h$: $h \in \Gamma_x$ for all $x \in X$ since $h$ is a self-loop of the hub and a neighbour of every rim vertex, so $\Omega_h = X$;
\item Output $i$ (rim): $i \in \Gamma_h$, $\Gamma_{i-1}$, $\Gamma_i$, $\Gamma_{i+1}$, so $|\Omega_i| \geq 4 > 1$.
\end{itemize}
Since no output exclusively identifies any single input, $\mathrm{ci}_\circ(N) = 0$.

\noindent (iii) By Theorem~\ref{theo1}, it suffices to exhibit a proper $\beta$-coloring of $\mathtt{W}_n$ such that within each 
color class every input has a private output.

\noindent{\bf Case 1} ($n$ even, $\beta = 3$). Color the rim vertices by parity:
\begin{equation*}
\mathcal{P}(i) = i\pmod{2},\quad \mathcal{P}(i) \in \{1, 2\}
\end{equation*}
and assign the hub a third color $\mathcal{P}(h) = 3$. Since $n$ is even, this is a proper $2$-coloring of the rim $\mathtt{C}_n$, and the hub receives a distinct third color. The color classes are $A = \{1, 3, \ldots, n-1\}$, $B = \{2, 4, \ldots, n\}$, and $C = \{h\}$. Within class $A$, consecutive vertices are at distance $2$ in $\mathtt{C}_n$, so their output-ranges $\Gamma_i = \{h, i-1, i, i+1\}$ and $\Gamma_{i+2} = \{h, i+1, i+2, i+3\}$ share only the hub output $h$ and $i+1$. The self-loop output $i \in \Gamma_i$ satisfies $i \notin \Gamma_{i+2}$ for $i \neq i+2$, so $i$ is private to vertex $i$ within $A$. The same holds for class $B$ by symmetry. The hub $h$ is alone in class $C$, so it is trivially 
identifiable.

\noindent{\bf Case 2} ($n$ odd, $\beta = 4$). Color the rim with $3$ colors as 
\begin{align*}
\mathcal{P}(i) \equiv 
i \pmod{3},\qquad\mathcal{P}(i) \in \{1, 2, 3\}, 
\end{align*}
which is a proper $3$-coloring of the odd cycle $\mathtt{C}_n$, and assign $\mathcal{P}(h) = 4$. Within each rim color class, consecutive same-color vertices are at distance $3$ in $\mathtt{C}_n$, so their output-ranges are disjoint and every vertex has a private self-loop output. The hub is alone in color class $4$ and trivially identifiable. Hence all $n+1$ inputs are conclusively identifiable.\\

\noindent (iv) Since $\mathrm{ci}_\circ(\mathrm{id}^c_\beta) = \beta$, 
the superactivation gap is:
\begin{equation*}
\Delta(N)=n + 1 - \beta = 
\begin{cases} 
n - 2 & \text{if } n \text{ is even,} \\ 
n - 3 & \text{if } n \text{ is odd,}
\end{cases}
\end{equation*}
which grows without bound as $n \to \infty$. This completes the proof. 
\end{proof}

\noindent Examples of other families of channels exhibiting such an unbounded advantage includes channels with their support graph as the star graphs $\mathtt{S}_n$, friendship graphs $\mathtt{F}_n$, and Turán graphs $\mathtt{T}(n,r)$ for fixed $r$ {(see Appendix {\bf A2})}

\section{Quantum Assistance in Conclusive Identification}\label{sec4}

A natural question is whether quantum resources can provide an advantage over classical assistance in the conclusive identification task. By Theorem~\ref{theo1}, an SNFC channel $N : X \to X$ with $|X| = n$ and $\mathrm{ci}_\circ(N) = 0$ requires a noiseless classical channel $\mathrm{id}^c_\beta$ with $\beta = \chi(\mathtt{S}_N)$ to conclusively identify all inputs in $X$. We are thus interested in the following question. 

\begin{question}\label{qu3}
Does there exist an SNFC channel $N : X \to X$ with $|X| = n$ and $\mathrm{ci}_\circ(N) = 0$, such that $\mathrm{ci}_\circ(N \otimes \mathrm{id}^q_{\beta})= n$ even with $\beta<\chi(\mathtt{S}_N)$?
\end{question}

\noindent \noindent In other words, can a noiseless quantum channel $\mathrm{id}^q_\beta$ of dimension strictly smaller than the chromatic number $\chi(\mathtt{S}_N)$ accomplish the conclusive identification task that requires a classical channel $\mathrm{id}^c_{\chi(\mathtt{S}_N)}$ by Theorem~\ref{theo1}? An affirmative answer  would demonstrate a strict quantum advantage in conclusive identification task.

\noindent Given access to a noiseless quantum channel $\mathrm{id}^q_d : \mathcal{D}(\mathbb{C}^d) \to \mathcal{D}(\mathbb{C}^d)$, assisting a SNFC channel $N : X \to X$, the sender, upon receiving an input $x \in X$, encodes it into a quantum state $\ket{v_x} \in \mathbb{C}^d$ and transmits this state through $\mathrm{id}^q_d$, while simultaneously sending $x$ through the channel $N$. Upon receiving the classical output $y \in X$ from $N$ together with the quantum system, the receiver performs a measurement that may depend on $y$, with the objective of conclusively identifying the input $x$. A natural question is: what is the minimum dimension $d$ of the assisting quantum channel that suffices for this task? To address this we recall the definition of orthogonal representation of a graph. 

\begin{definition}[Orthogonal Representation \& Rank]
Let $\mathtt{G} = (V,E)$ be a finite simple graph. An orthogonal representation of $G$ in  $\mathbb{C}^d$ is a map $\phi : V \to \mathbb{C}^d \setminus \{0\}, \quad v \mapsto \ket{v}$ such that $\langle u \mid v \rangle = 0~ \text{whenever } (u,v) \in E$. The orthogonal rank of $\mathtt{G}$, denoted by $\xi(\mathtt{G})$, is defined as $$\xi(\mathtt{G}) := \min \left\{ d \in \mathbb{N}~:~ \exists \text{ an orthogonal representation of } \mathtt{G} \text{ in } \mathbb{C}^d \right\}.$$
\end{definition}

\noindent With this, we are now ready to present our first result in the presence of quantum assistance.

\begin{theorem}\label{theo4}
Let $N : X \to X$ be an SNFC channel with $|X| = n$, satisfying $c_\circ(N)=0$ and $\mathrm{ci}_\circ(N)=0$, and let $\mathtt{S}_N$ denote its support graph. Then a noiseless quantum channel $\mathrm{id}^q_{\beta}$ with $\beta = \xi(\mathtt{S}_N)$ is sufficient to achieve $\mathrm{ci}_\circ\!\left(N \otimes \mathrm{id}^q_{\beta}\right) = n$. 
\end{theorem}
\begin{proof}
Let $\{ \ket{v_x} \}_{x \in X}$ be an orthogonal representation of $\mathtt{S}_N$ in $\mathbb{C}^{\beta}$, where $\beta = \xi(\mathtt{S}_N)$. The protocol is defined as follows.

\noindent \emph{Encoding:} Upon input $x \in X$, the sender prepares the state $\ket{v_x}$ and transmits it through the noiseless quantum channel $\mathrm{id}^q_{\beta}$, while simultaneously sending $x$ through the classical channel $N$.

\noindent \emph{Decoding:} Upon receiving the classical output $y \in X$, the receiver performs the two-outcome projective measurement
\begin{align*}
\left\{ \ket{v_y}\!\bra{v_y},\; I - \ket{v_y}\!\bra{v_y} \right\}   
\end{align*}
on the received quantum state, and declares YES if the first outcome occurs and NO otherwise.

\noindent\emph{Correctness:} If $y = x$, then the received state is $\ket{v_x} = \ket{v_y}$, and hence the outcome corresponding to $\ket{v_y}\!\bra{v_y}$ occurs with probability $1$, leading to a correct YES output. If $y \neq x$ and $P(y|x) > 0$, then $(x,y) \in E(\mathtt{S}_N)$ by definition of the support graph. By the orthogonality of the representation, $\langle v_y \mid v_x \rangle = 0$, and therefore the outcome corresponding to $\ket{v_y}\!\bra{v_y}$ occurs with probability $|\langle v_y \mid v_x \rangle|^2 = 0$. Hence the receiver outputs NO with certainty. Thus, the protocol achieves $\mathrm{ci}_\circ(N \otimes \mathrm{id}^q_{\beta}) = n$. By minimality of $\beta = \xi(\mathtt{S}_N)$, the claim follows.
\end{proof}

\noindent With Theorems~\ref{theo1} and~\ref{theo4} in hand, any SNFC channel $N$ with support graph $\mathtt{S}_N$ satisfying $\xi(\mathtt{S}_N) < \chi(\mathtt{S}_N)$ provides an affirmative answer to Question~\ref{qu3}. Explicit examples of such channels are given in the next subsection.  

\subsection{KS Advantage in Conclusive Identification}

\noindent Here we first digress a bit to recall fundamental result from quantum foundations -- the Kochen-Specker (KS) no-go theorem~\cite{Kochen1967}, which places constraints on the possibility of assigning definite values to quantum observables prior to measurement.\\

\noindent{Kochen-Specker No-Go Result}\vspace{-.25cm}\\

\begin{figure}[t!]
\centering
\includegraphics[width=0.5\textwidth]{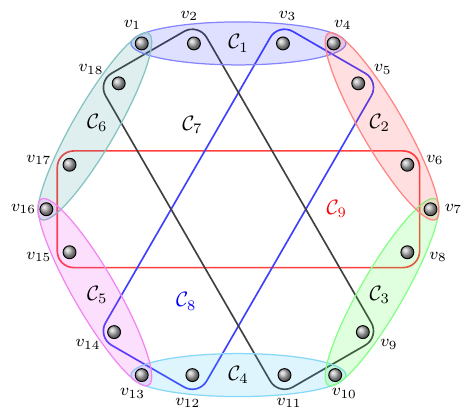}
\vspace{1em}
\renewcommand{\arraystretch}{1}
\begin{tabular}{c c c}
\hline
$\scriptstyle \mathcal{C}_1=\{v_1,v_2,v_3,v_4\}$ &
$\scriptstyle \mathcal{C}_2=\{v_4,v_5,v_6,v_7\}$ &
$\scriptstyle \mathcal{C}_3=\{v_7,v_8,v_9,v_{10}\}$ \\

$\scriptstyle \mathcal{C}_4=\{v_{10},v_{11},v_{12},v_{13}\}$ &
$\scriptstyle \mathcal{C}_5=\{v_{13},v_{14},v_{15},v_{16}\}$ &
$\scriptstyle \mathcal{C}_6=\{v_{16},v_{17},v_{18},v_{1}\}$\\

$\scriptstyle \mathcal{C}_7=\{v_{2},v_{9},v_{11},v_{18}\}$ &
$\scriptstyle \mathcal{C}_8=\{v_{3},v_{5},v_{12},v_{14}\}$ &
$\scriptstyle \mathcal{C}_9=\{v_{6},v_{8},v_{15},v_{17}\}$\\
\hline\hline
\end{tabular}
\caption{The hypergraph depicting the $18$-vector KS system in $\mathbb{C}^4$ \cite{Cabello1996}. Each vertex $v_i$ (shown in Fig.~\ref{fig7}) corresponds to a unit vector in $\mathbb{C}^4$. Each shaded region corresponds to one of the $9$ orthonormal bases $B_1, \ldots, B_9$, each corresponding to different hyper-edge (contexts) $\mathcal{C}_1, \ldots, \mathcal{C}_9$. The absence of a valid $\{0,1\}$-coloring, assigning exactly one $1$ per context, is guaranteed by the parity argument since $9$ is odd and each vector appears in exactly $2$ contexts, thereby constituting the KS proof of contextuality.}
\label{fig6}
\end{figure}

\noindent A central question in quantum foundations is whether quantum mechanics admits a `hidden variable' description \cite{Einstein1935,Bell1964,Bell1966,Mermin1993}. Can the outcomes of all possible measurements on a quantum system be predetermined by some underlying classical state, independent of the measurement context? A value assignment is a map $\mu: \mathcal{O} \to \{0, 1\}$ from a set of observables $\mathcal{O}$ to definite values $\{0,1\}$, satisfying:
\begin{itemize}[itemsep=-0.15cm, parsep=0.15cm, topsep=2pt, leftmargin=0.4cm]
\item \emph{Non-contextuality:} for every $\mathrm{P}\in\mathcal{O}$, $\mu(\mathrm{P})$ is independent of which other compatible observables are measured simultaneously;
\item \emph{Completeness:} for every orthonormal basis $\{e_1, \ldots, e_d\}$ of $\mathcal{H}$, exactly one projector $\ket{e_i}\bra{e_i}$ is assigned value $1$ and the rest are assigned $0$.
\end{itemize}
Such a map is called a \emph{non-contextual hidden variable} (NCHV) assignment. The Kochen-Specker (KS) theorem establishes that no such assignment exists for Hilbert spaces of dimension $d \geq 3$ \cite{Kochen1967} (see also \cite{Budroni2022}). The obstruction can be captured combinatorially through the notion of a \emph{Kochen-Specker system}.

\begin{definition}[KS System]\label{def11}
A KS system in dimension $d$ is a pair $(\mathcal{V}, \mathcal{B})$, where $\mathcal{V} = \{v_1, \ldots, v_m\}$ is a finite set of unit vectors in $\mathbb{C}^d$ and $\mathcal{B} = \{B_1, \ldots, B_k\}$ is a collection of orthonormal bases (called \emph{contexts}), with $\mathcal{V} = \bigcup_j B_j$, such that no $\{0,1\}$-coloring of $\mathcal{V}$ satisfies: (i) Exclusivity (at most one vector in each context $B_j$ is colored $1$) and (ii) Completeness (exactly one vector in each context $B_j$ is colored $1$).
\end{definition}

\noindent The non-existence of such a coloring is precisely what makes the KS system a proof of contextuality. The smallest known KS system in dimension $d = 3$ consists of $31$ vectors~\cite{Peres1995} (see \cite{Kernaghan2026}), while in dimension $d = 4$ a system of as few as $18$ vectors is known~\cite{Cabello1996} (see Fig.~\ref{fig6}).

\begin{definition}[KS Graph]\label{def12}
The \emph{orthogonality graph} $\mathtt{KS}_{|\mathcal{V}|}$ of a KS system $(\mathcal{V}, \mathcal{B})$ is an undirected graph with vertex set $\mathcal{V}$, where two vertices $v_i$ and $v_j$ are adjacent if and only if $v_i \perp v_j$. A graph $\mathtt{KS}^d$ is called a KS graph in dimension $d$ if it is the orthogonality graph of some KS system in $\mathbb{C}^d$.
\end{definition}

\noindent For instance, the KS graph $\mathtt{KS}^4_{18}$ corresponding to the $18$-vector KS system in $\mathbb{C}^4$~\cite{Cabello1996} is depicted in Fig.~\ref{fig7}. With this we now provide an affirmative answer to Question~\ref{qu3} through an explicit construction.

\begin{theorem}\label{theo5}
There exists an SNFC channel $N : X \to X$ with $c_\circ(N) = 0$, $\mathrm{ci}_\circ(N) = 0$, and $\chi(\mathtt{S}_N) < |X|$, such that $\mathrm{ci}_\circ(N \otimes \mathrm{id}^q_d) = |X|$ for some $d < \chi(\mathtt{S}_N)$, demonstrating a strict quantum advantage over classical assistance in the conclusive identification task.
\end{theorem}

\begin{proof}
Consider the SNFC channel $N : X \to X$ with $|X| = 18$ and support graph $\mathtt{S}_N = \mathtt{KS}^4_{18}$. Since $\mathrm{diam}(\mathtt{KS}^4_{18}) = 2$ (see Fig.~\ref{fig7}), Proposition~\ref{prop3} gives $\mathtt{G}_N = K_{18}$, and hence $c_\circ(N) = 0$. Since every output $y \in X$ satisfies $|\Omega_y| > 1$, no input is conclusively identifiable without assistance, giving $\mathrm{ci}_\circ(N) = 0$. By Theorem~\ref{theo1}, since $\chi(\mathtt{KS}^4_{18}) = 5$ (see Fig.~\ref{fig7}), the minimum classical assistance required for conclusive identification of all $18$ inputs is $\mathrm{id}^c_5$, i.e.~$\mathrm{ci}_\circ(N \otimes \mathrm{id}^c_5) = 18$. On the other hand, since orthogonal rank $\xi(\mathtt{KS}^4_{18}) = 4$, a $4$-dimensional noiseless quantum channel $\mathrm{id}^q_4$ suffices. Following the protocol of Theorem~\ref{theo4}, one obtains
\begin{equation*}
\mathrm{ci}_\circ(N \otimes \mathrm{id}^q_4) = 18, \qquad  4 = \xi(\mathtt{KS}^4_{18}) < 5 = \chi(\mathtt{KS}^4_{18}),
\end{equation*}
establishing a strict quantum advantage in the conclusive identification task.
\end{proof}

\begin{observation}\label{obs4}
Given the hypergraph of Fig.~\ref{fig6}, construct a graph $\tilde{\mathtt{KS}}^4_{18} := (\mathcal{V}, \mathcal{E})$ by replacing each hyperedge (context) with a clique: for each context $h_j = \{v_{i_1}, v_{i_2}, v_{i_3}, v_{i_4}\}$, all $6$ edges $\{v_{i_k} \sim v_{i_\ell} \mid k \neq \ell\}$ are added to $\mathcal{E}$. The resulting graph $\tilde{\mathtt{KS}}^4_{18}$ has the same $18$ vertices as $\mathtt{KS}^4_{18}$, but fewer edges: since each of the $9$ contexts contributes $6$ edges, $\tilde{\mathtt{KS}}^4_{18}$ has $9 \times 6 = 54$ edges, whereas $\mathtt{KS}^4_{18}$ has $63$ edges. Thus $\tilde{\mathtt{KS}}^4_{18}$ is a proper subgraph of $\mathtt{KS}^4_{18}$: $    \tilde{\mathtt{KS}}^4_{18} \subsetneq \mathtt{KS}^4_{18}$. Notably, the claim of Theorem~\ref{theo5} remains valid when the support graph of the SNFC channel $N$ is taken to be $\mathtt{S}_N = \tilde{\mathtt{KS}}^4_{18}$ rather than $\mathtt{KS}^4_{18}$.
\end{observation}

\noindent The observation above reveals an important subtlety: the quantum advantage established in Theorem~\ref{theo5} is not a consequence of the full edge structure of $\mathtt{KS}^4_{18}$, but rather of its underlying \emph{context structure} encoded in the hypergraph of Fig.~\ref{fig6}. Indeed, $\tilde{\mathtt{KS}}^4_{18}$ is a strictly sparser graph than $\mathtt{KS}^4_{18}$ -- having $9$ fewer edges -- yet admits the same classical $5$-coloring. This suggests that the context structure of a KS system, rather than the precise orthogonality graph, is the fundamental ingredient driving quantum advantage in the conclusive identification task.

\begin{figure}[t]
\centering
\includegraphics[width=0.55\textwidth]{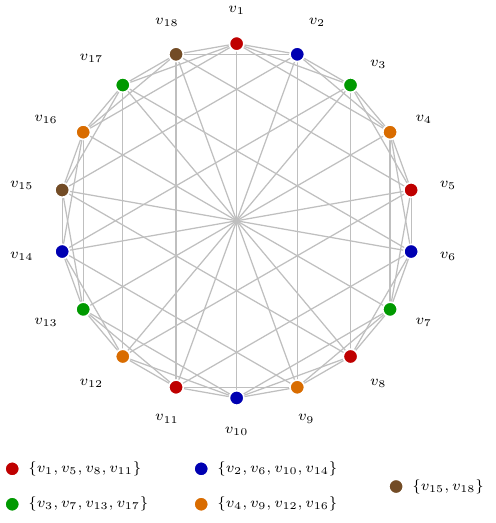}
\renewcommand{\arraystretch}{.9}
\begin{tabular}{c c c c}
\hline
$\scriptstyle v_1=(1,0,0,0)$ &
$\scriptstyle v_2=(0,1,0,0)$ &
$\scriptstyle v_3=(0,0,1,1)$ &
$\scriptstyle v_4=(0,0,1,-1)$ \\

$\scriptstyle v_5=(1,-1,0,0)$ &
$\scriptstyle v_6=(1,1,-1,-1)$ &
$\scriptstyle v_7=(1,1,1,1)$ &
$\scriptstyle v_8=(1,-1,1,-1)$ \\

$\scriptstyle v_9=(1,0,-1,0)$ &
$\scriptstyle v_{10}=(0,1,0,-1)$ &
$\scriptstyle v_{11}=(1,0,1,0)$ &
$\scriptstyle v_{12}=(1,1,-1,1)$ \\

$\scriptstyle v_{13}=(-1,1,1,1)$ &
$\scriptstyle v_{14}=(1,1,1,-1)$ &
$\scriptstyle v_{15}=(1,0,0,1)$ &
$\scriptstyle v_{16}=(0,1,-1,0)$ \\

\multicolumn{4}{c}{$\scriptstyle v_{17}=(0,1,1,0)\hspace{1em} v_{18}=(0,0,0,1)$} \\
\hline\hline
\end{tabular}
\caption{The $\mathtt{KS}^4_{18}$ graph. Each vertex $v_i:= a_i\ket{0} + b_i\ket{1} + c_i\ket{2} + d_i\ket{3}$, denoted in shorthand as $v_i = (a_i, b_i, c_i, d_i)$, corresponds to a unit vector in $\mathbb{C}^4$. Two vertices are adjacent if and only if the corresponding vectors are orthogonal. Every vertex has degree $7$ (except the self-loop). The colors indicate a proper $5$-coloring, confirming $\chi(\mathtt{KS}^4_{18}) = 5$.}
\label{fig7}
\end{figure}

\subsection{Beyond KS (system) advantage in Conclusive Identification}

\subsubsection{Advantage Through State Independent Contextuality of Yu and Oh}

\noindent Recently, Yu and Oh~\cite{Yu2012} provided a novel proof of quantum contextuality in $\mathbb{C}^3$ using $13$ rays. Like the standard KS proofs, their proof is state-independent; however, the construction is \emph{not} a KS system in the sense of Definition~\ref{def11}. Among the $13$ vectors in $\mathbb{R}^3$, there exist only $4$ mutually orthogonal triples, and a valid $\{0,1\}$-coloring satisfying both completeness and exclusivity can be explicitly constructed. The Yu-Oh proof instead establishes contextuality through an operator-algebraic identity for spin-$1$ observables -- a mechanism fundamentally different from the combinatorial coloring argument underlying the standard KS proofs. The orthogonality graph $\mathtt{YO}^3_{13}$ of the Yu-Oh construction is depicted in Fig.~\ref{fig7}. Despite not being a KS system in the sense of Definition~\ref{def11}, we show in what follows that $\mathtt{YO}^3_{13}$ nevertheless gives rise to a quantum advantage in the conclusive identification task.

\begin{theorem}\label{theo6}
The SNFC channel $N : X \to X$ with support graph  $\mathtt{S}_N=\mathtt{YO}^3_{13}$ has $c_\circ(N) = 0$, $\mathrm{ci}_\circ(N) = 0$, and demonstrates a strict quantum advantage over classical assistance in the conclusive identification task.
\end{theorem}
\begin{proof}
Since $\mathrm{diam}(\mathtt{YO}^3_{13})=2$, in accordance with Proposition~\ref{prop3}, the confusability graph of $N: X \to X$ is $\mathtt{G}_N = \mathtt{K}_{13}$, giving $c_\circ(N) = 0$. Since every output $x \in X$ satisfies $|\Omega_x| > 1$, no input can be conclusively identified without assistance, and hence $\mathrm{ci}_\circ(N) = 0$.

\noindent Since $\chi(\mathtt{YO}^3_{13}) = 4$ (see Fig.~\ref{fig8}), Theorem~\ref{theo1} implies that the minimum classical assistance required for conclusive identification of all $13$ inputs is $\mathrm{id}^c_4$, i.e. $\mathrm{ci}_\circ(N \otimes \mathrm{id}^c_4) = 13$.

\noindent On the other hand, according to Theorem~\ref{theo4}, a $3$-dimensional noiseless quantum channel $\mathrm{id}^q_3$ suffices the purpose. Thus we have $\mathrm{ci}_\circ(N \otimes \mathrm{id}^q_3) = 13$, thereby establishing a strict quantum advantage over the optimal classical assistance.
\end{proof}

\noindent The above result, establishes that 
quantum advantage in the conclusive identification task is not an isolated phenomenon tied to the KS systems only. 

\subsubsection{Advantage Through State Dependent Contextuality}

\noindent The non-existence of a $\{0,1\}$-coloring satisfying \emph{exclusivity} and \emph{completeness} for KS system ((Definition~\ref{def11}), or the violation of an operator-algebraic identity -- as in the Yu-Oh proof -- establishes a contradiction with \emph{non-contextual hidden variable} (NCHV) models that holds for \emph{every} quantum state in the Hilbert space, regardless of its preparation. The contradiction is therefore a property of the algebraic or combinatorial structure of the observables alone, not of the quantum state. Such proofs are called \emph{state-independent} (SI) proofs of contextuality.

\noindent In contrast, a \emph{state-dependent} (SD) contextuality proof establishes a contradiction with NCHV assignments only for a particular quantum state $\ket{\psi}$, or a restricted class of states. The contradiction arises through the violation of a \emph{non-contextuality inequality}, a bound that holds for all NCHV models, but is violated by the quantum expectation value for the specific state $\ket{\psi}$. The celebrated Bell-CHSH inequality~\cite{Bell1964,Bell1966,Mermin1990} is the paradigmatic example: it is violated by entangled states such as $\ket{\Phi^+}$, but not by separable states. The Klyachko-Can-Binicio\u{g}lu-Shumovsky (KCBS) inequality~\cite{Klyachko2008} provides another well-known instance, being violated only for states outside the classical polytope — one of the first SD contextuality proofs in a single qutrit system.

\noindent While both SI and SD proofs witness the failure of NCHV models, they differ in scope: SI proofs are stronger, as the non-classicality is intrinsic to the measurement structure and holds universally, whereas SD proofs reveal non-classicality only for specific state-measurement pairs. In what follows, we show that a strict quantum advantage in the conclusive identification task can be obtained from SD contextuality as well.

\begin{theorem}\label{theo7}
There exists an SNFC channel $N : X \to X$ whose support graph $\mathtt{S}_N$ arises from a state-dependent (SD) contextuality construction, such that $\mathrm{ci}_\circ(N \otimes \mathrm{id}^q_d) = |X|~\text{for~} d < \chi(\mathtt{S}_N)$, establishing a strict quantum advantage over the optimal classical assistance $\mathrm{id}^c_{\chi(\mathtt{S}_N)}$ in the conclusive identification task.
\end{theorem}
\begin{proof}
Consider the $13$-vector set $\mathcal{V}_{13} \subset\mathbb{C}^3$ of the Yu and Oh's construction~\cite{Yu2012}. Append a new vector $v_{14}$ that is orthogonal to all $13$ vectors in $\mathcal{V}_{13}$. The resulting $14$-vector set $\mathcal{V}_{14} := \mathcal{V}_{13} \cup \{v_{14}\} \subset \mathbb{C}^4$ and its orthogonality graph $\mathtt{YO}^4_{14}$ are well-defined. The vertex $v_{14}$ plays a role analogous to the hub vertex of a wheel graph: just as the hub of $W_n$ is adjacent to every rim vertex, $v_{14}$ is adjacent to every vertex in $\mathtt{YO}^3_{13}$. Accordingly, $\mathrm{diam}(\mathtt{YO}^4_{14})=2$, and the chromatic number satisfies $\chi(\mathtt{YO}^4_{14}) = \chi(\mathtt{YO}^3_{13}) + 1 = 5$. While $\mathcal{V}_{14}$ gives rise to a state-dependent (SD) proof of contextuality, it does not constitute a state-independent (SI) proof, as established in Theorem~1 of~\cite{Cabello2015}.

\noindent Now consider an SNFC channel $N : X \to X$ with $|X| = 14$ and support graph $\mathtt{S}_N =\mathtt{YO}^4_{14}$. By Proposition~\ref{prop3}, the confusability graph is $\mathtt{G}_N = K_{14}$, giving $c_\circ(N) = 0$. Since $|\Omega_x| > 1$ for all $x \in X$, no input is conclusively identifiable without assistance, so $\mathrm{ci}_\circ(N) = 0$. By Theorem~\ref{theo1}, the minimum classical assistance required is $\mathrm{id}^c_5$, since $\chi(\mathtt{YO}^4_{14}) = 5$. On the other hand, Theorem~\ref{theo4} ensures $\mathrm{ci}_\circ(N \otimes \mathrm{id}^q_4)= 14$, thereby establishing a strict quantum advantage arising from an SD contextuality construction.
\end{proof}

\noindent In summary, for an SNFC channel $N : X \to X$, with $c_\circ(N)=0~\&~\mathrm{ci}_\circ(N)=0$, the minimum classical assistance required for conclusive identification is $\mathrm{id}^c_{\chi(\mathtt{S}_N)}$, determined by the chromatic number of the support graph $\mathtt{S}_N$. The quantum advantage arises whenever a noiseless quantum channel $\mathrm{id}^q_d$ with $d < \chi(\mathtt{S}_N)$ suffices to accomplish the same task. The key quantity governing this is the \emph{orthogonal rank} $\xi(\mathtt{S}_N)$ of the support graph (in $\mathbb{C}^d$). The three explicit examples constructed in this section demonstrate that this separation can arise from qualitatively different types of contextuality: (i) Combinatorial state-independent contextuality (Comb. SIC), (ii) Algebraic state-independent contextuality (Alg. SIC), and (iii) State-dependent contextuality (SDC). The three examples are summarized in Table~\ref{tab:quantum-advantage}.

\begin{figure}[t]
\centering
\includegraphics[width=0.45\textwidth]{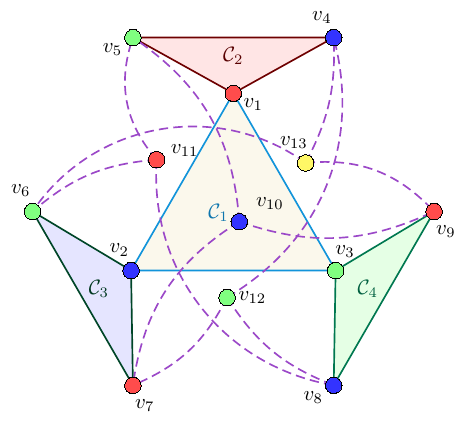}
\renewcommand{\arraystretch}{1}
\begin{tabular}{c c c c c}
\hline
$\scriptstyle v_1=(1,0,0)$ &
$\scriptstyle v_2=(0,1,0)$ &
$\scriptstyle v_3=(0,0,1)$ &
$\scriptstyle v_4=(0,1,1)$\\

$\scriptstyle v_5=(0,1,-1)$ &
$\scriptstyle v_6=(1,0,1)$ &
$\scriptstyle v_7=(1,0,-1)$ &
$\scriptstyle v_8=(1,1,0)$ \\

$\scriptstyle v_9=(1,-1,0)$ &
$\scriptstyle v_{10}=(1,1,1)$ &
$\scriptstyle v_{11}=(-1,1,1)$ &
$\scriptstyle v_{12}=(1,-1,1)$ \\
\multicolumn{4}{c}{$\scriptstyle v_{13}=(1,1,-1)$} \\
\hline
$\scriptstyle \mathcal{C}_1=\{v_1,v_2,v_3\}$ &
$\scriptstyle \mathcal{C}_2=\{v_1,v_4,v_5\}$ &
$\scriptstyle \mathcal{C}_3=\{v_2,v_6,v_7\}$ &
$\scriptstyle \mathcal{C}_4=\{v_3,v_8,v_9\}$ \\\hline\hline
\end{tabular}
\caption{ Yu-Oh's $13$-ray algebraic state-independent proof of contextuality in $\mathbb{C}^3$ \cite{Yu2012}. This is not a KS System as per Definition~\ref{def11}, as $\bigcup_{i=1}^4\mathcal{C}_i$ does not contain all the $13$ vectors $\mathcal{V}_{13}$. Precisely, the vectors $\{v_{10},v_{11},v_{12},v_{13}\}$ do not appear in any of the context. The graph $\mathtt{YU}^3_{13}$ has $\xi(\mathtt{YU}^3_{13})=3$, $\mathrm{diam}(\mathtt{YU}^3_{13})=2$, and $\chi(\mathtt{YU}^3_{13})=4$ (a proper $4$-coloring is depicted). }
\label{fig8}
\end{figure}

\begin{table}[h]
\centering
\begin{tabular}{|c||c|c|c|c|c|}
\hline
Construction & ~Type~ & ~$|X|$~ & $\chi(\mathtt{S}_N)$ & $\xi(\mathtt{S}_N)$ & ~Gap~ \\\hline\hline
$\mathtt{KS}^4_{18}$~\cite{Cabello1996} & Comb.\ SIC & $18$ & $5$ & $4$ & $1$ \\\hline $\mathtt{YO}^3_{13}$~\cite{Yu2012} & Alg.\ SIC & $13$ & $4$ & $3$ & $1$ \\\hline $\mathtt{YO}^4_{14}$~\cite{Cabello2015} & SDC & $14$ & $5$ & $4$ & $1$ \\\hline
\end{tabular}
\caption{Summary of quantum advantages. For each SNFC channel $N$ with support graph $\mathtt{S}_N$, the gap $\chi(\mathtt{S}_N) -\xi(\mathtt{S}_N) = 1$, in all three cases.}
\label{tab:quantum-advantage}
\end{table}

\noindent Taken together, these results demonstrate that the separation $\xi(\mathtt{S}_N) < \chi(\mathtt{S}_N)$ is the fundamental criterion for quantum advantage in the conclusive identification task. A natural question is how large the quantum advantage can be; we address this in the next subsection.

\subsection{Scaling the Quantum Advantage}

In this subsection we present families of SNFC channels for which the quantum advantage in the conclusive identification task -- measured by the ratio of the classical to quantum assistance required -- can be scaled arbitrarily large. 

\subsubsection{Polynomial Scaling of Quantum Advantage} We start by recalling a key algebraic tool from graph theory called the \emph{co-normal product}~\cite{Feige1997}.

\begin{definition}[Co-Normal Product]\label{def:conormal}
The \emph{co-normal product} (also known as the \emph{disjunctive product}) of two graphs $\mathtt{G} = (V(\mathtt{G}), E(\mathtt{G}))$ and $\mathtt{H} = (V(\mathtt{H}), E(\mathtt{H}))$ is the graph $\mathtt{G} \times \mathtt{H}$ with vertex set $V(\mathtt{G}) \times V(\mathtt{H})$, where two distinct vertices $(i, j)$ and $(i', j')$ are adjacent if and only if
\begin{equation*}
(i,j) \sim (i',j') \iff i \sim_{\mathtt{G}} i' \quad \lor \quad j \sim_{\mathtt{H}} j'.
\end{equation*}
We write $\mathtt{G}^{\times m}$ for the $m$-fold co-normal product of $\mathtt{G}$ with itself.
\end{definition}

\noindent The following properties of the co-normal product are relevant to our analysis:
\begin{itemize}[itemsep=-0.15cm, parsep=0.15cm, 
               topsep=2pt, leftmargin=1.0cm]
\item[\textbf{(P1)}] If $\mathrm{diam}(\mathtt{G}) \leq 2$, then $\mathrm{diam}(\mathtt{G}^{\times m}) \leq 2$ for all $m \in \mathbb{N}$~\cite{Lovasz1979, 
    Hammack2011}.
\item[\textbf{(P2)}] If $\xi(\mathtt{G}) = d$, then $\xi(\mathtt{G}^{\times m}) \leq d^m$ for all $m \in \mathbb{N}$~\cite{Hogben2017}.
\item[\textbf{(P3)}] $\chi_f(\mathtt{G}^{\times m}) = \chi_f(\mathtt{G})^m$ for all $m \in \mathbb{N}$~\cite{Lovasz1979, Scheinerman1997, Hammack2011}.\footnote{For any graph $\mathtt{G} = (V, E)$, the \emph{fractional chromatic number} $\chi_f(\mathtt{G})$ is a relaxation of the chromatic number satisfying $\omega(\mathtt{G}) \leq \chi_f(\mathtt{G}) \leq \chi(\mathtt{G})$, where $\omega(\mathtt{G})$ denotes the clique number of $\mathtt{G}$.}
\end{itemize}

\begin{definition}[Co-Normal Product 
Channel]\label{def:conormal-channel} Given an SNFC channel $N : X \to X$ with support graph $\mathtt{S}_N$, the \emph{$m$-fold co-normal product channel} $N^{\times m} : X^m \to X^m$ is the channel whose support graph is the $m$-fold co-normal product $\mathtt{S}_N^{\times m}$, where $X^m$ denotes the $m$-fold Cartesian product of the alphabet $X$ with itself, i.e. $X^m := \underbrace{X \times \cdots \times X}_{m \text{ times}}$.
\end{definition}

\begin{theorem}\label{theo8}
Let $N : X \to X$ be an SNFC channel with support graph $\mathtt{S}_N$ satisfying:
\begin{itemize}[itemsep=-0.15cm, parsep=0.15cm, topsep=2pt, leftmargin=1.0cm]
\item[\emph{a)}] $\mathrm{diam}(\mathtt{S}_N) = 2$, so that $\mathtt{G}_N = K_{|X|}$ and $c_\circ(N) = 0$;
\item[\emph{b)}] $\mathrm{ci}_\circ(N) = 0$;
\item[\emph{c)}] $\chi(\mathtt{S}_N) < |X|$, so that classical superactivation is possible by Theorem~\ref{theo1};
\item[\emph{d)}] The \emph{quantum advantage ratio}
\begin{equation*}
\mathtt{QA}(N) := \frac{\chi(\mathtt{S}_N)}{\xi(\mathtt{S}_N)}\ge\frac{\chi_f(\mathtt{S}_N)}{\xi(\mathtt{S}_N)} > 1,
\end{equation*}
so that quantum assistance strictly outperforms classical assistance.
\end{itemize}
Then for the $m$-fold co-normal product channel $N^{\times m}$, the quantum advantage ratio $\mathtt{QA}(N^{\times m})$ can be made arbitrarily large as $m \to \infty$.
\end{theorem}

\begin{proof}
Since $\mathrm{diam}(\mathtt{S}_N)=2$, it follows that 
\begin{align*}
\mathrm{diam}(\mathtt{S}_N^{\times m}) \le 2 \qquad (\text{property } \mathbf{P1}).   
\end{align*}
Hence, by Proposition~\ref{prop3}, the confusability graph of the co-normal product channel $N^{\times m}$ is the complete graph $\mathtt{K}_{|X|^m}$, which implies $c_\circ(N^{\times m}) = 0$. 

\noindent Further, since $\mathrm{ci}_\circ(N)=0$, no input symbol is conclusively identifiable without assistance. This property is preserved under the co-normal product, and therefore $\mathrm{ci}_\circ(N^{\times m}) = 0$.

\noindent Now, the quantum advantage ratio for $N^{\times m}$ satisfies
\begin{align*}
\mathtt{QA}(N^{\times m})&= \frac{\chi(\mathtt{S}_N^{\times m})}{\xi(\mathtt{S}_N^{\times m})}
\;\ge\; \frac{\chi_f(\mathtt{S}_N^{\times m})}{\xi(\mathtt{S}_N^{\times m})} \\
&\ge\; \frac{\big(\chi_f(\mathtt{S}_N)\big)^m}{\big(\xi(\mathtt{S}_N)\big)^m}
\qquad [\text{by } \mathbf{P2} \text{ and } \mathbf{P3}].
\end{align*}
Since, by assumption, $\frac{\chi_f(\mathtt{S}_N)}{\xi(\mathtt{S}_N)} > 1$, it follows that
\begin{align*}
\mathtt{QA}(N^{\times m}) \;\ge\; \left(\frac{\chi_f(\mathtt{S}_N)}{\xi(\mathtt{S}_N)}\right)^m,
\end{align*}
which grows unboundedly with $m$. Hence, the quantum advantage for the $m$-fold co-normal product channel can be made arbitrarily large as $m \to \infty$.
\end{proof}

\noindent For an explicit instantiation of Theorem~\ref{theo8}, one may consider any SNFC channel whose support graph is $\mathtt{KS}^{4}_{18}$, $\mathtt{YO}^{3}_{14}$, or $\mathtt{YO}^{4}_{14}$, as studied in Theorems~\ref{theo5}, \ref{theo6}, and \ref{theo7}, respectively. For such channels, the quantum advantage ratio in Theorem~\ref{theo8}, obtained via the $m$-fold co-normal product construction, can be made arbitrarily large by increasing $m$. Notably, in this construction, the  assistance of noiseless classical channel grows polynomially as compared to the assistance of noiseless quantum channel. In the following, we demonstrate that there exist channels for which the required classical resource scales exponentially compared to its quantum counterpart.      

\subsubsection{Exponential Scaling of Quantum Advantage}

We begin by recalling the Newman graph \cite{Newman}, which plays a central role in demonstrating an exponential separation between classical and quantum resources.

\begin{definition}[Newman states and Newman graph]
Let $d = 4k$ for some $k \in \mathbb{N}$. Define
\begin{align*}
\mathcal{N}_d := \left\{ x \in \{\pm 1\}^d \;\middle|\; \sum_{i=1}^d x_i \equiv 0 \pmod{4} \right\},
\end{align*}
i.e., the set of all $\{\pm 1\}$-valued vectors in $\mathbb{R}^d$ having an even number of $-1$ entries. The elements of $\mathcal{N}_d$ are referred to as \emph{Newman states}.

The \emph{Newman graph} $\mathtt{Y}_d=(V(\mathtt{Y}_d),E(\mathtt{Y}_d))$ is the orthogonality graph on $\mathcal{N}_d$, i.e., $V(\mathtt{Y}_d)=\mathcal{N}_d$, where two distinct vertices $x,y \in \mathcal{N}_d$ are adjacent if and only if $\langle x, y \rangle = 0$.
\end{definition}

\noindent We will first prove a key property of the Newman graph ensuring its strong connectivity which will be instrumental for our purpose.

\begin{lemma}[Diameter of the Newman graph]\label{lemma2}
Let $d>4$ be a multiple of $4$. Then $\operatorname{diam}(\mathtt{Y}_d)=2$.
\end{lemma}
\begin{proof}
For $x,y \in \{\pm 1\}^d$, let $H(x,y)$ denote the Hamming distance. Then
\begin{align*}
\langle x, y \rangle = d - 2H(x,y),   
\end{align*}
and hence $\langle x, y \rangle = 0$ if and only if $H(x,y)=d/2$. Thus, two vertices are adjacent in $\mathtt{Y}_d$ if and only if $H(x,y)=d/2$. Fix two vertices $x,y \in \mathcal{N}_d$, and define
\begin{align*}
A=\{i: x_i = y_i\}, \qquad B=\{i: x_i \neq y_i\},
\end{align*}
with $|A|=a$ and $|B|=b$. Since both $x$ and $y$ contain an even number of $-1$ entries, it follows that $b$ is even, and hence $a=d-b$ is also even. If $b=d/2$, then $H(x,y)=d/2$, and thus $x$ and $y$ are adjacent. Otherwise, suppose $b \ne d/2$. We construct a vector $z \in \{\pm 1\}^d$ as follows:
\begin{itemize}
\item On the coordinates in $A$, flip exactly $a/2$ entries of $x$.
\item On the coordinates in $B$, let $z$ agree with $x$ on exactly $b/2$ coordinates and with $y$ on the remaining $b/2$ coordinates.
\end{itemize}
By construction,
\begin{align*}
H(x,z)=\frac{a}{2}+\frac{b}{2}=\frac{d}{2}, \qquad
H(y,z)=\frac{a}{2}+\frac{b}{2}=\frac{d}{2},
\end{align*}
so $z$ is adjacent to both $x$ and $y$. Moreover, since $a/2$ and $b/2$ are integers and $d/2$ is even (as $d=4k$), the parity condition is preserved, ensuring that $z \in \mathcal{N}_d$. Thus any two vertices are either adjacent or share a common neighbor, which implies $\mathrm{diam}(\mathtt{Y}_d) \le 2$. Finally, for $d>4$, there exist non-adjacent vertices (e.g., vertices with $H(x,y)=2 \neq d/2$), so $\mathrm{diam}(\mathtt{Y}_d) \neq 1$. Therefore, $\mathrm{diam}(\mathtt{Y}_d)=2$. This completes the proof.
\end{proof}
With this we now present the main result of this subsection. 

\begin{theorem}
Let $d=4k$ for some $k\in\mathbb{N}$, and consider an SNFC channel $N_d:X\to X$ with $|X|=|V(Y_d)|=2^{d-2}$ whose support graph is the Newman graph, i.e., $\mathtt{S}_{N_d}=Y_d$. Then the following properties hold:
\begin{itemize}[itemsep=-0.15cm, parsep=0.15cm, topsep=2pt, leftmargin=1.0cm]
\item[\emph{(i)}] $c_\circ(N_d)=0$;
\item[\emph{(ii)}] $\mathrm{ci}_\circ(N_d)=0$;
\item[\emph{(iii)}] $\mathtt{QA}(N_d)=\frac{\chi(\mathtt{Y}_d)}{\xi(\mathtt{Y}_d)} \;\ge\; \frac{1}{d}\left(\frac{2}{1.99}\right)^d$, thereby exhibiting an exponential quantum advantage in $d$.
\end{itemize}
\end{theorem}
\begin{proof}
(i) Since $\operatorname{diam}(Y_d)=2$ (Lemma~\ref{lemma2}), it follows from Proposition~\ref{prop3} that the confusability graph of the channel $N_d$ is the complete graph $\mathtt{K}_n$ on $n=2^{d-2}$ vertices, and hence $c_\circ(N_d)=0$.

\noindent (ii) Since $|\Omega_y|>1$ for every $y\in X$, no input symbol is conclusively identifiable without assistance, which implies $\mathrm{ci}_\circ(N_d)=0$.

\noindent (iii) By construction, $\xi(Y_d)=d$. Hence, from Theorem~\ref{theo4}, it follows that
\begin{align*}
\mathrm{ci}_\circ\!\left(N_d\otimes \mathrm{id}^q_d\right)=|X|=2^{d-2}.    
\end{align*}
On the other hand, the independence number of the Newman graph satisfies
\begin{align*}
\alpha(Y_d)\le \frac{(1.99)^d}{4}.
\end{align*}
This follows from the bound $\alpha(\mathtt{H}_d)\le (1.99)^d$ for the Hadamard graph $\mathtt{H}_d$~\cite{Frankl1987}, together with the relation $\alpha(Y_d)=\alpha(H_d)/4$~\cite{Newman}. Using the general bound $\chi(G)\ge |V(G)|/\alpha(G)$, we obtain
\begin{align*}
\chi(Y_d)\;\ge\;\frac{|V(Y_d)|}{\alpha(Y_d)}
\;\ge\;\frac{2^{d-2}}{(1.99)^d/4}
\;=\;\left(\frac{2}{1.99}\right)^d.
\end{align*}
Therefore,
\begin{align*}
\mathtt{QA}(N_d)=\frac{\chi(Y_d)}{\xi(Y_d)}
\;\ge\;\frac{1}{d}\left(\frac{2}{1.99}\right)^d,
\end{align*}
which establishes an exponential quantum advantage.
\end{proof}
Theorem~\ref{theo8} thus provides an explicit family of channels exhibiting exponential quantum advantage. This advantage can further be amplified by considering $m$-fold co-normal product channel of Newman channels.

\section{Discussion}\label{sec:discussion}

\noindent{\it Summary of Results}\vspace{-.35cm}\\ 

\noindent In this work we introduced the conclusive identification task for classical channels and established a comprehensive theory connecting it to graph theory, quantum contextuality, and information-theoretic superactivation. The central message is that the \emph{support graph} $\mathtt{S}_N$, rather than the confusability graph $\mathtt{G}_N$, is the appropriate combinatorial object for characterizing channel utility in this task. The key results can be summarized as a chain of implications. For an SNFC channel $N : X \to X$ with support graph $\mathtt{S}_N$ and having $c_\circ(N)=0~\&~\mathrm{ci}_\circ(N)=0$:
\begin{itemize}[itemsep=-0.15cm, parsep=0.15cm, topsep=2pt, leftmargin=1.0cm]
\item[$\checkmark$] $\chi(\mathtt{S}_N) = \beta < |X| \implies \mathrm{ci}_\circ(N \otimes \mathrm{id}^c_\beta) = |X|$ establishes a new form of activation;  
\item[$\checkmark$] $\xi(\mathtt{S}_N) < \chi(\mathtt{S}_N) \implies \mathrm{ci}_\circ(N \otimes \mathrm{id}^q_{\xi}) = |X|$ establishes an quantum advantage; 
\item[$\checkmark$] $\mathtt{QA}(N) > 1 \implies \mathtt{QA}(N^{\times m}) \geq \mathtt{QA}(N)^m$ establishes scaling of quantum advantage;
\item[$\checkmark$] $\mathtt{QA}(N_d)\geq \frac{1}{d}\left(\frac{2}{1.99}\right)^d$ establishes an exponential quantum advantage.
\end{itemize}
This quantum advantage is particularly noteworthy. The seminal result of Holevo establishes that, in the absence of shared entanglement, a noiseless quantum channel cannot transmit more classical information than its classical counterpart~\cite{Holevo1973}. More recently, this limitation has been extended to the single-shot regime, where it has been shown that, even in the presence of shared classical randomness, any channel simulable by a $d$-dimensional quantum system can also be simulated exactly by a $d$-level classical system~\cite{Frenkel2015}. Together, these results imply that within Shannon’s communication framework, a noiseless quantum channel offers no intrinsic advantage over a classical one. In contrast, quantum superdense coding protocol demonstrates that, in the presence of shared entanglement, a quantum channel can outperform its classical counterpart by a factor of two~\cite{Bennett1992}. The quantum advantage established here reveals a qualitatively different phenomenon: even without any additional quantum resources, a noiseless quantum channel can serve as a strictly more powerful form of assistance than a classical noiseless channel, yielding, in fact, an exponential advantage. \\    

\noindent{\it Role of Quantum Contextuality}\vspace{-.35cm}\\ 

\noindent A central finding of this paper is that quantum contextuality, in all its manifestations, is a resource for quantum advantage in the conclusive identification task. The three explicit constructions presented cover the full spectrum of contextuality proofs:
\begin{itemize}[itemsep=-0.15cm, parsep=0.15cm, topsep=2pt, leftmargin=1.0cm]
\item[$\checkmark$] Combinatorial state-independent contextuality ($\mathtt{KS}^4_{18}$): the impossibility of a $\{0,1\}$-coloring directly witnesses the separation $\xi(\mathtt{KS}^4_{18}) = 4 < 5 =\chi(\mathtt{KS}^4_{18})$.
\item[$\checkmark$] Algebraic state-independent contextuality ($\mathtt{YO}^3_{13}$): an operator-algebraic identity for spin-$1$ observables witnesses the separation $\xi(\mathtt{YO}^3_{13}) = 3 < 4 = \chi(\mathtt{YO}^3_{13})$.
\item[$\checkmark$] State-dependent contextuality    ($\mathtt{YO}^4_{14}$): a construction that is not state-independent nevertheless witnesses $\xi(\mathtt{YO}^4_{14}) = 4 < 5 = \chi(\mathtt{YO}^4_{14})$.
\end{itemize}
Note that a set of rank-one projectors with orthogonality graph $\mathtt{G}$ constitutes a state-independent contextuality proof if $\xi(\mathtt{G})<\chi_f(\mathtt{G})$, which immediately yields $\xi(\mathtt{G}) < \chi(\mathtt{G})$~\cite{Cabello2015}. Therefore, SNFC channels with support graph same as orthogonality graph of any combinatorial or algebraic state-independent contextuality proof exhibits quantum advantage in conclusive identification task. State-dependent contextuality constructions need not satisfy this condition in general, but those that do, such as $\mathtt{YO}^4_{14}$, leads to similar quantum advantage. Moreover, if these graphs have diameter two and have no isolated vertex, then the resulting channels have $c_\circ(N) = 0$ as well as $\mathrm{ci}_\circ(N) = 0$. This connection parallels known results linking contextuality to advantages in measurement-based quantum computation~\cite{Raussendorf2013}, magic state distillation~\cite{Howard2014}, and communication complexity~\cite{Gupta2023,Chowdhury2026}, and suggests that the conclusive identification task provides a new and natural operational arena for probing quantum contextuality.\\

\noindent{\it Open Problems}\vspace{-.25cm}\\ 

\noindent Several natural questions remain open and, we believe, merit further investigation.

\begin{itemize}[itemsep=-0.15cm, parsep=0.15cm, topsep=2pt, leftmargin=1.0cm]
\item[1)] \textbf{Asymptotic conclusive identification index.} The asymptotic index can be defined as $\mathrm{Ci}_\circ(N) = \lim_{k \to \infty} \frac{1}{k} \log_2 \mathrm{ci}_\circ(N^{\otimes k})$. It would be interesting to study this asymptotic quantity. 
\item[2)] \textbf{Larger quantum advantage gaps.} All three explicit examples achieve a gap of $\chi(\mathtt{S}_N) - \xi(\mathtt{S}_N) = 1$. Can larger gaps be achieved by other KS systems or contextuality constructions, and what is the maximum achievable gap for a given $|X|$?
\item[3)] \textbf{Tight characterization of quantum advantage.} Theorem~\ref{theo8} gives a lower bound on $\mathtt{QA}(N^{\times m})$. Is $\mathtt{QA}(N^{\times m}) = \mathtt{QA}(N)^m$ exactly, or can the bound be improved?
\item[4)] \textbf{Entanglement-assisted conclusive identification.} We considered assistance by noiseless classical and quantum channels. What is the role of shared entanglement in the conclusive identification task? Can entanglement play nontrivial role beyond the superdense coding protocol?
\item[5)] \textbf{Quantum channels.} The conclusive identification task was defined for classical channels. A natural extension is to quantum channels $\mathcal{N} : \mathcal{L}(\mathcal{H}) \to \mathcal{L}(\mathcal{K})$, where the sender encodes classical messages into quantum states and the receiver performs a POVM. Does a quantum analogue of the support graph exist, and does it give rise to analogous superactivation and quantum advantage phenomena?
\item[6)] \textbf{Bounded error variant of the task.} While we have studied the conclusive identification task, an interesting variant could be its bounded error variant where small error is allowed even when the receiver guess the inputs correctly. Such variant is well studied in communication complexity \cite{Gavinsky2006} and learning problem \cite{Regev2009}.   
\end{itemize}

\noindent{\it Concluding Remarks}\vspace{-.35cm}\\

\noindent The conclusive identification task introduced in this paper occupies a natural position at the intersection of zero-error information theory, graph theory, and quantum foundations. It reveals that classical channels with $C_\circ(N) = 0$ -- universally dismissed as useless by Shannon's zero-error framework -- can be activated by noiseless classical channel assistance to achieve perfect conclusive identification, with the required assistance determined precisely by the chromatic number of the support graph. The further activation by quantum assistance, governed by the orthogonal rank of the support graph, provides a new and direct operational manifestation of quantum contextuality in a classical communication setting. We hope that the framework introduced here will stimulate further investigation at the interface of information theory and quantum foundations.

\section*{Acknowledgments}
\noindent We thankfully acknowledge Jaikumar Radhakrishnan for pointing out the possible connection between support graph and confusability graph (private communication). MB acknowledges the financial support through the National Quantum Mission (NQM) of the Department of Science and Technology, Government of India.

%
%
%
%
%


%
%
%
%
%
%
%
\newpage 
\section*{Appendix}
\subsection*{{\bf A1.} Identification scheme for SNFC channels with support graph $\mathtt{S}_i$ of Eq.(\ref{penta})}\label{Appendix-A}
\noindent While the graph $\mathtt{S}_4$ in Eq.~(\ref{penta}) is $3$-colorable, for $\mathtt{S}_6$ a proper coloring requires $4$ different colors. The same coloring shown in Fig.~\ref{fig4} for the graph $\mathtt{S}_5$ in-fact works for $\mathtt{S}_6$. With this we now detail the scheme as of Table~\ref{tab1} for channels with support graphs $\mathtt{S}_4$ and $\mathtt{S}_6$ in Table~\ref{app:tab3}.

\begin{table}[h!]
\centering
\footnotesize
\setlength{\tabcolsep}{3pt}
\renewcommand{\arraystretch}{1.08}

\begin{minipage}[t]{0.48\textwidth}
\centering
\resizebox{\linewidth}{!}{%
\begin{tabular}{c|c|c|c}
\hline
Input of $N$ & Output of $N$ & Message via $\mathrm{id}^c_3$ & Identified? \\
\hline\hline
\multirow{5}{*}{$1$} 
& $1$ & R & $\checkmark$ \\
& $2$ & R & $\checkmark$ \\
& $3$ & R & $\checkmark$ \\
& $4$ & R & $\checkmark$ \\
& $5$ & R & $\checkmark$ \\
\hline
\multirow{4}{*}{$2$} 
& $1$ & B & $\times$ \\
& $2$ & B & $\checkmark$ \\
& $3$ & B & $\times$ \\
& $5$ & B & $\times$ \\
\hline
\multirow{4}{*}{$3$}
& $1$ & G & $\times$ \\
& $2$ & G & $\times$ \\
& $3$ & G & $\checkmark$ \\
& $4$ & G & $\times$ \\
\hline
\multirow{4}{*}{$4$}
& $1$ & B & $\times$ \\
& $3$ & B & $\times$ \\
& $4$ & B & $\checkmark$ \\
& $5$ & B & $\times$ \\
\hline
\multirow{4}{*}{$5$}
& $1$ & G & $\times$ \\
& $2$ & G & $\times$ \\
& $4$ & G & $\times$ \\
& $5$ & G & $\checkmark$ \\
\hline
\end{tabular}%
}
\vspace{0.5ex}

{\footnotesize (a) $\mathtt{S}_N=\mathtt{S}_4$}
\end{minipage}
\hfill
\begin{minipage}[t]{0.48\textwidth}
\centering
\resizebox{\linewidth}{!}{%
\begin{tabular}{c|c|c|c}
\hline
Input of $N$ & Output of $N$ & Message via $\mathrm{id}^c_4$ & Identified? \\
\hline\hline
\multirow{5}{*}{$1$} 
& $1$ & R & $\checkmark$ \\
& $2$ & R & $\checkmark$ \\
& $3$ & R & $\checkmark$ \\
& $4$ & R & $\checkmark$ \\
& $5$ & R & $\checkmark$ \\
\hline
\multirow{5}{*}{$2$} 
& $1$ & B & $\checkmark$ \\
& $2$ & B & $\checkmark$ \\
& $3$ & B & $\checkmark$ \\
& $4$ & B & $\checkmark$ \\
& $5$ & B & $\checkmark$ \\
\hline
\multirow{4}{*}{$3$}
& $1$ & G & $\times$ \\
& $2$ & G & $\times$ \\
& $3$ & G & $\checkmark$ \\
& $4$ & G & $\times$ \\
\hline
\multirow{5}{*}{$4$}
& $1$ & M & $\checkmark$ \\
& $1$ & M & $\checkmark$ \\
& $3$ & M & $\checkmark$ \\
& $4$ & M & $\checkmark$ \\
& $5$ & M & $\checkmark$ \\
\hline
\multirow{4}{*}{$5$}
& $1$ & G & $\times$ \\
& $2$ & G & $\times$ \\
& $4$ & G & $\times$ \\
& $5$ & G & $\checkmark$ \\
\hline
\end{tabular}%
}
\vspace{0.5ex}

{\footnotesize (b) $\mathtt{S}_N=\mathtt{S}_6$}
\end{minipage}

\caption{The conclusive identification scheme. (a) An SNFC channel $N$ with support graph $\mathtt{S}_{N}=\mathtt{S}_4$. Partition: $\mathcal{P}(1)=\mathrm{R},~\mathcal{P}(2)=\mathcal{P}(4)=\mathrm{B},~\mathcal{P}(3)=\mathcal{P}(5)=\mathrm{G}$. (b) An SNFC channel $N$ with support graph $\mathtt{S}_{N}=\mathtt{S}_6$. Partition: $\mathcal{P}(1)=\mathrm{R},~\mathcal{P}(2)=\mathrm{B},~\mathcal{P}(3)=\mathcal{P}(5)=\mathrm{G},~\mathcal{P}(4)=\mathrm{M}$.}
\label{app:tab3}
\end{table}

\subsection*{{\bf A2.} Other classes of channels exhibiting Arbitrarily Large Superactivation}

\noindent Theorem~\ref{theo3} shows that for an SNFC channel $N_n : X \to X$ with support graph $\mathtt{S}_{N_n} = \mathtt{W}_n$, the superactivation gap $\Delta(N_n) := \mathrm{ci}_\circ(N_n \otimes \mathrm{id}^c_\beta) - \mathrm{ci}_\circ(\mathrm{id}^c_\beta)$ grows unboundedly, with $\Delta(N_n)=n-2$ for odd $n$ and $\Delta(N_n)=n-3$ for even $n$. Such unbounded superactivation is not restricted to wheel graphs alone, but also arises for channels whose support graphs belong to other families. In particular, we highlight three additional classes: the Friendship graph $\mathtt{F}_n$ (see Fig.~\ref{fig:friendship}), the Star graph $\mathtt{St}_n$ (see Fig.~\ref{fig:star}), and the Turán graph (see Fig.~\ref{fig:turan_n2}).   
\begin{figure}[h]
\centering
\includegraphics[width=0.45\textwidth]{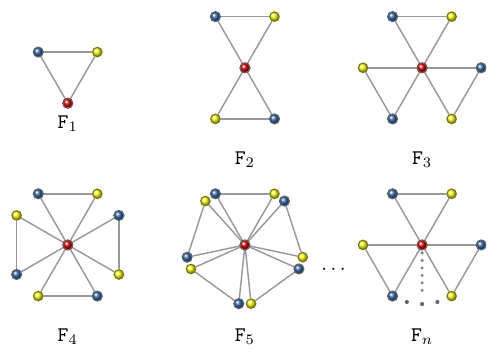}
\caption{Friendship graph $\mathtt{F}_n$~\cite{Longyear1972}. This family of graphs is motivated by the classical friendship theorem~\cite{Erdos1966}, which states that if every pair of vertices has exactly one common neighbor, then there exists a unique universal vertex adjacent to all others. The resulting structure consists of a collection of triangles sharing a single common vertex.}
\label{fig:friendship}
\end{figure}
\begin{figure}[h]
\centering
\includegraphics[width=0.45\textwidth]{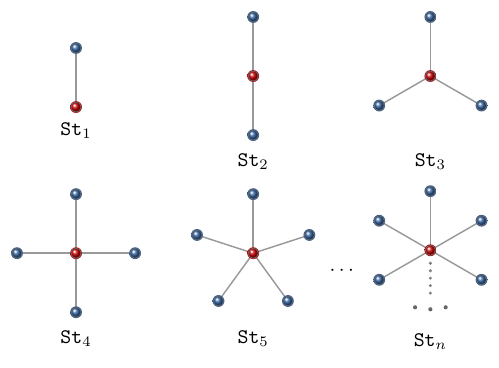}
\caption{Star graph $\mathtt{St}_n$~\cite{harary1969}. It represents one of the simplest extremal structures in graph theory, consisting of a single universal vertex connected to $n$ leaves, with no edges among the leaves.}
\label{fig:star}
\end{figure}
\begin{figure}[h!]
\centering
\includegraphics[width=0.45\textwidth]{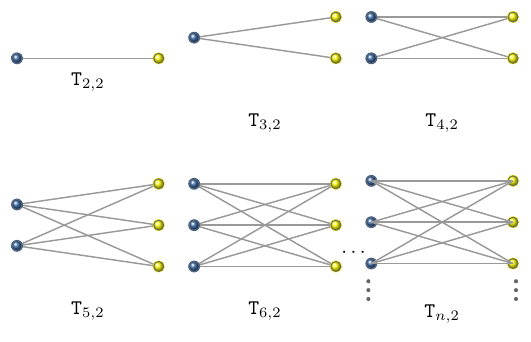}
\caption{Turán graph $\mathtt{T}_{n,2}$~\cite{turan1941}. It is the complete bipartite graph whose two parts are as equal in size as possible. It arises from Turán's extremal theorem as the unique graph maximizing the number of edges among all triangle-free graphs on $n$ vertices. When $n$ is even, $\mathtt{T}_{n,2}$ reduces to the balanced complete bipartite graph, also known as the cocktail party graph~\cite{Roberts1969}.}
\label{fig:turan_n2}
\end{figure}
\noindent Several properties of these graphs and properties of the SNFC channels having these as support graphs are listed in Table~\ref{tab:list}.  
\clearpage
\begin{table}[!t]
\centering
\begin{tabular}{|c||c|c|c|c|c|c|}
\hline
& \multicolumn{2}{c|}{$\mathtt{F}_n$} & \multicolumn{2}{c|}{$\mathtt{St}_n$}  & \multicolumn{2}{c|}{$\mathtt{T}_{n,2}$} \\ \hline\hline
 & odd & even & odd & even & odd & even \\ \hline
$\mathrm{diam}$ & $2^{\textcolor{blue}{\star}}$ & $2$ & $2^{\textcolor{blue}{\star}}$ & $2$ & $2$ & $2^{\textcolor{blue}{\star}}$ \\ \hline
Chrom. No. $\chi$ & $3$ & $3$ & $2$ & $2$ & $2$ & $2$ \\ \hline
Ind. No. $\alpha$  & $3$ & $3$ & $2$ & $2$ & $2$ & $2$ \\ \hline
$c_\circ(N_n)$ & 0 & 0 & 0 & 0 & 0 & 0 \\ \hline
$\mathrm{ci}_\circ(N_n)$ & 0 & 0 & 0 & 0 & 0 & 0 \\ \hline
$\mathrm{ci}_\circ(N_n\otimes\mathrm{id}^c_\chi)$ & $2n$+$1$ & $2n$+$1$ & $n$+$1$ & $n$+$1$ & $n$ & $n$ \\ \hline
$\Delta(N_n)$ & $2n$-$2$ & $2n$-$2$ & $n$-$1$ & $n$-$1$ & $n$-$2$ & $n$-$2$ \\ \hline
\end{tabular}
\caption{Properties of Friendship, Star, and Turán graphs. [${\textcolor{blue}{\star}}$]: $\mathrm{diam}=1$ for $\mathtt{F}_1$, $\mathtt{St}_1$, and $\mathtt{T}_{2,2}$. Superactivation phenomenon is obtained for all these classes of channels and the superactivation gap can be arbitrarily large as $n\to\infty$.}
\label{tab:list}
\end{table}

\end{document}